\documentclass[11pt]{article}

\usepackage[T1]{fontenc}
\usepackage[utf8]{inputenc}
\usepackage{lmodern}
\usepackage{amsmath,amssymb,amsthm,mathtools,bm}
\usepackage{booktabs,array,enumitem}
\usepackage{graphicx}
\usepackage{placeins}
\usepackage[margin=1in]{geometry}
\usepackage{xcolor}
\usepackage{microtype}
\usepackage{etoc}
\usepackage{tikz}
\usetikzlibrary{arrows.meta,positioning,fit}
\usepackage[colorlinks=true,allcolors=blue!55!black]{hyperref}
\hypersetup{
  pdftitle={When Expressivity Is Not Enough: Discrete Routing Geometry in Variational Quantum Circuits},
  pdfauthor={Yu Wang},
  pdfsubject={Discrete routing geometry, tangent accessibility, and same-point architecture repair},
  pdfkeywords={variational quantum algorithms, CNOT routing,
  operator Schmidt rank, tangent space, quantum architecture search}
}

\newtheorem{theorem}{Theorem}[section]
\newtheorem{proposition}[theorem]{Proposition}
\newtheorem{lemma}[theorem]{Lemma}
\newtheorem{corollary}[theorem]{Corollary}
\theoremstyle{definition}
\newtheorem{definition}[theorem]{Definition}

\theoremstyle{remark}

\newcommand{\Ftwo}{\mathbb F_2}
\newcommand{\GL}{\mathrm{GL}}
\newcommand{\SU}{\mathrm{SU}}
\newcommand{\PU}{\mathrm{PU}}
\newcommand{\Tr}{\operatorname{Tr}}
\newcommand{\rank}{\operatorname{rank}}
\newcommand{\OSR}{\operatorname{OSR}}

\newcommand{\im}{\operatorname{im}}
\newcommand{\Span}{\operatorname{span}}
\newcommand{\cT}{\mathcal T}
\newcommand{\cD}{\mathcal D}
\newcommand{\cP}{\mathcal P}
\newcommand{\cE}{\mathcal E}
\newcommand{\cL}{\mathcal L}

\newcommand{\ket}[1]{\lvert #1\rangle}
\newcommand{\bra}[1]{\langle #1\rvert}
\newcommand{\braket}[2]{\langle #1\vert #2\rangle}

\newcommand{\norm}[1]{\left\lVert #1\right\rVert}
\newcommand{\abs}[1]{\left\lvert #1\right\rvert}
\newcommand{\inner}[2]{\left\langle #1,#2\right\rangle}
\newcommand{\barS}{\overline S}
\newcommand{\id}{\mathbb I}

\title{\bfseries When Expressivity Is Not Enough: Discrete Routing Geometry in Variational Quantum Circuits}
\author{Yu Wang}
\date{}

\begin{document}
\etocdepthtag.toc{main}
\begin{center}
{\LARGE\bfseries When Expressivity Is Not Enough: Discrete Routing Geometry in Variational Quantum Circuits\par}
\vspace{1em}
{\large Yu Wang\par}
\vspace{0.4em}
{\small Hetao Institute of Mathematics and Interdisciplinary Sciences,\\
Shenzhen, Guangdong 518017, China\par}
\vspace{0.3em}
{\small \href{mailto:ming-jing-happy@163.com}{ming-jing-happy@163.com}\par}
\end{center}
\vspace{0.5em}

\begin{abstract}
Finding useful variational quantum circuits requires more than
representational capacity: training must also have access to directions
that improve the objective. We study how discrete CNOT routing shapes
both capacity and local accessibility.
For circuits built from CNOT and arbitrary one-qubit gates, we separate
binary routing matrices from continuous rotations and use this
representation to relate global capacity to local tangent accessibility.
An exact finite-field formula for the operator Schmidt rank of a CNOT
block, together with cumulative multiblock bounds, yields capacity
constraints and a Bell-pair fidelity obstruction.
Locally, routing propagates generators into a routing-dependent tangent
space, so all parameter gradients can vanish even when the ambient
objective gradient is nonzero; frame bounds further separate loss of
gradient coverage from poor conditioning.
This viewpoint motivates a same-point augmentation: a routing-conditioned
identity insertion preserves the current unitary, state, loss, and
existing gradients while adding new tangent directions.
Candidate-only and incremental projection scores provide local criteria
for ranking such insertions. Exact-statevector tests show strong
short-horizon score--descent correlations on controlled Bell tasks and
effective equal-cost selection on connected Heisenberg chains up to
twelve qubits, with complementary TFIM selector controls. A broader TFIM
routing ensemble exhibits routing-dependent near-stagnation.
Together, these results connect representational capacity, local
accessibility, and same-point repair through a common discrete routing
description. An architecture need not remain fixed during training:
task-dependent geometric information can guide the opening of new descent
directions while preserving the learned circuit at insertion. This provides
a concrete local approach to a broader problem: how useful quantum circuits
can be discovered from the information a task supplies.
\end{abstract}

\section{Introduction}

Quantum algorithms demonstrate the potential of quantum computation to
surpass classical methods.  Shor's factoring algorithm, Grover's search,
the HHL algorithm for linear systems, and quantum simulation show how
coherent evolution and interference can support computational advantages
under appropriate problem and resource assumptions
\cite{Shor1997,Grover1996,harrow2009quantum,Lloyd1996}.  Their central
structures, however, emerged from specific insights into number theory,
search, linear algebra, or quantum dynamics: researchers identified useful
structure and organized it into quantum circuits.  These achievements
demonstrate the power of explicit design while raising a further question:
can useful quantum computational structures be discovered systematically
through task feedback and computational search?

Classical deep learning offers an instructive comparison.  Given data,
objectives, and computational resources, training can develop complex
representations without an explicitly programmed rule for every step of a
task \cite{LeCun2015DeepLearning}.  Computation thus participates in finding
effective representations as well as executing them.  Variational quantum
algorithms, quantum neural networks, and quantum architecture search seek
to bring this approach to quantum computation
\cite{Peruzzo2014,Farhi2014,Kandala2017,Mitarai2018,Du2022QAS,Nemkov2023}.
They express tasks through energies, fidelities, or data-dependent losses
and search for circuits through feedback between quantum execution and
classical optimization.  Human insight and computational search can
complement each other, but search introduces an additional requirement:
a useful circuit must be found within an affordable resource budget.

This ambition faces substantial obstacles.  Performance depends strongly
on the ansatz, often chosen from experience, hardware templates, or random
constructions.  Hardware connectivity, finite sampling, and noise constrain
execution and optimization.  Barren plateaus and related mechanisms can
suppress gradients under conditions involving circuit depth, expressivity,
cost locality, noise, and algebraic structure
\cite{McClean2018,Cerezo2021,Wang2021Noise,Holmes2022,Larocca2022,
Ragone2024,Fontana2024}.  Classical learning also offers a structural lesson:
continuous weights operate within architectures that shape information and
gradient flow, as illustrated by convolution, residual connections, and
attention \cite{LeCun1998,He2016,Vaswani2017}.  Such architectures may be
designed or searched \cite{Elsken2019}; increasing parameter count alone
does not replace choosing an appropriate structure.  Conventional
fixed-architecture quantum training likewise varies the angles in
$U(\bm\theta)$ while fixing the entangling structure.  Quantum architecture
search already opens some discrete choices, but a mechanistic question
remains: how do these choices determine which task directions continuous
training can access?

We address this question through the coupling between discrete CNOT routing
and continuous local control.  In its given gate order, a circuit over CNOT
and arbitrary one-qubit gates can be grouped into alternating CNOT-only
blocks and local layers.  Each CNOT block induces an invertible binary map
$A_\ell\in\GL(n,\Ftwo)$, while each local layer belongs to $\SU(2)^n$.
For $t$ routing blocks this gives the mixed encoding
\begin{equation}
 \cE_{n,t}
 =
 \bigl(\GL(n,\Ftwo)\bigr)^t
 \times
 \bigl(\SU(2)^n\bigr)^{t+1}.
 \label{eq:intro-encoding}
\end{equation}
Here $\bm A=(A_1,\ldots,A_t)$ specifies logical binary routing, distinct
from hardware SWAP routing, and $\bm\theta$ describes the local controls.
The encoding is redundant, and different routing sectors can have
overlapping unitary images.  Under standard compilation and accuracy
conventions, the logical-unitary parts of uniform polynomial-size quantum
circuit families admit such encodings with polynomially many routing blocks
\cite{Barenco1995,BernsteinVazirani1997,DawsonNielsen2006}.
Conversely, every binary routing map admits a polynomial-size CNOT
synthesis under all-to-all logical connectivity \cite{PatelMarkovHayes08}.
The union of these encodings over polynomial resource bounds therefore
provides a common representation of polynomial-size logical unitary
circuits, extending beyond any particular variational ansatz.
Explicit design constructs $(\bm A,\bm\theta)$ from problem structure;
variational discovery searches for such coordinates through task feedback.
Fixed-routing training explores only one continuous family within this
framework.  The relevant question is whether this restriction excludes
descent directions required by the task.

We show that discrete routing controls both global capacity and local
accessibility.  It determines the cross-cut operator Schmidt rank of a
CNOT block and propagates local generators into the tangent directions
exposed by continuous parameters.  A compact acceptable circuit can
therefore exist in the full mixed family while the current fixed-routing
circuit satisfies
\begin{equation}
 \Pi_{\cT_A}G=0,\qquad \norm{G}>0.
 \label{eq:intro-hook}
\end{equation}
Here $A$ denotes the fixed routing tuple, $\cT_A$ its accessible tangent
space at the current point, and $G$ the ambient Hermitian objective
gradient.  The physical objective admits descent, yet every exposed
first-order variation is orthogonal to it.  This deterministic local
obstruction requires neither random-parameter concentration nor noise.
We further construct a routing-conditioned augmentation initialized at
the identity.  It preserves the current unitary, output state, loss, and
existing gradients while adding tangent directions, restoring descent
when a suitable candidate overlaps the task gradient.
Figure~\ref{fig:hierarchy} illustrates this same-point test and repair of
architecture-dependent accessibility.

\begin{figure}[!htbp]
\centering
\includegraphics[width=\linewidth,trim=0 112bp 0 112bp,clip]{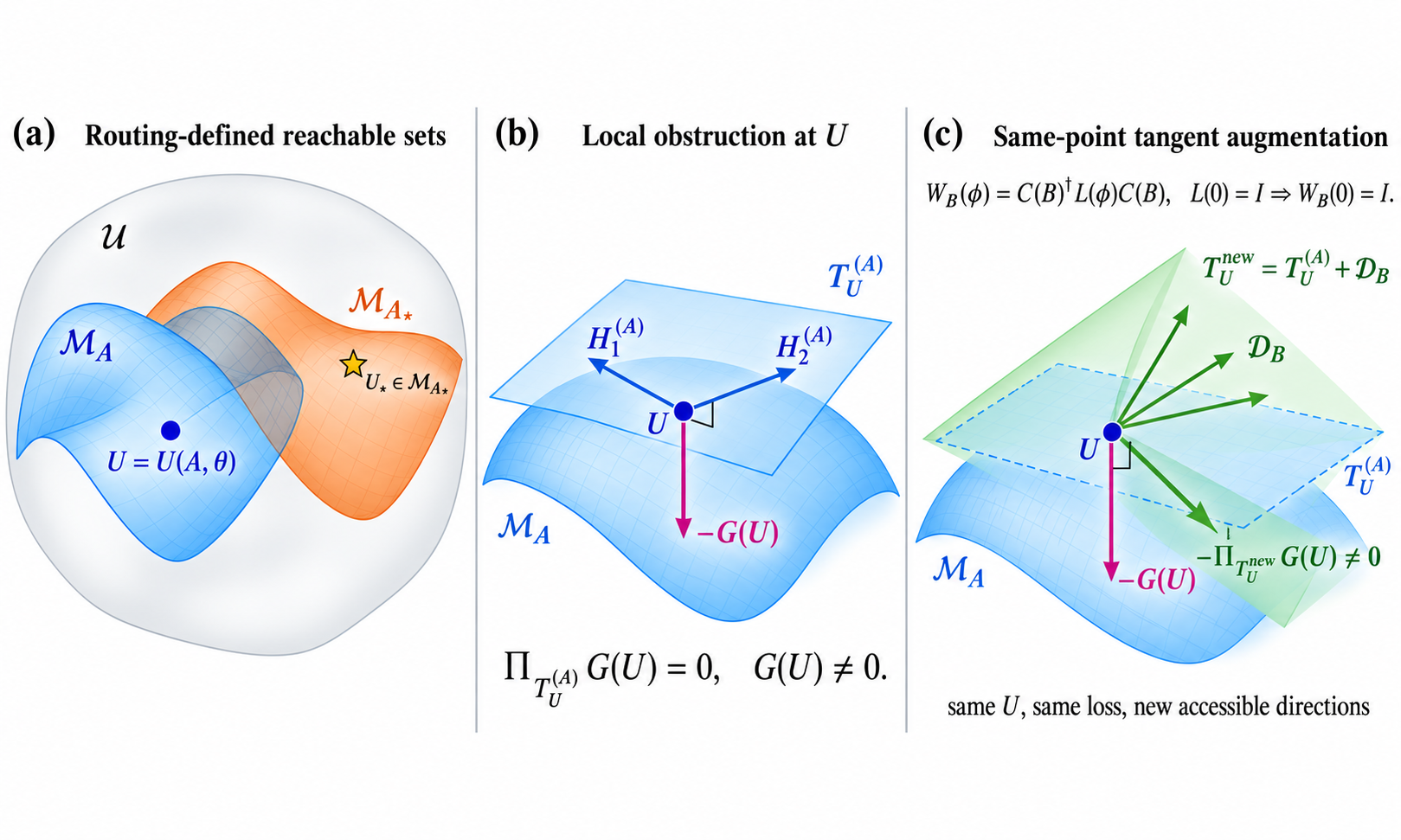}
\caption{\textbf{Routing, local accessibility, and same-point repair.}
(a) Varying the angles at fixed routing $A$ traces a reachable set
$\mathcal M_A$ within the ambient unitary space.  The target $U_\star$
belongs to $\mathcal M_{A_\star}$; multiblock reachable sets may overlap.
(b) At the current $U$, the ambient descent direction $-G(U)$ can be
orthogonal to the accessible space $\cT_U^{(A)}$, so parameter gradients
vanish despite $G(U)\ne0$.
(c) An identity-initialized routing block preserves $U$, its output state,
loss, and ambient gradient while adding directions $\cD_B$.
For a suitable $B$, the enlarged space
$\cT_U^{\rm new}=\cT_U^{(A)}+\cD_B$ exposes a nonzero descent projection.
Surfaces and shaded spans are schematic; surface height does not encode loss.}
\label{fig:hierarchy}
\end{figure}

To establish this mechanism, we derive an exact finite-field capacity
formula and cumulative bounds, prove tight Bell-pair fidelity and
accessibility obstructions, and separate direction coverage from parameter
conditioning through tangent-frame analysis.  An exact binary
state-tangent certificate at a CNOT origin supplies conditional local
gradient guarantees.  Controlled multiblock circuits, a connected
Heisenberg-chain task, and transverse-field Ising controls provide
finite-size numerical evidence under exact statevector access.  Together,
these results make discrete routing a computable, diagnosable, and locally
repairable design variable, distinguishing circuit representation from
local accessibility and resource-bounded discovery.  Potential quantum
advantage additionally requires accounting for discovery costs, hardware
execution, sampling resources, and appropriate classical baselines;
favorable gradients can coexist with classical tractability in studied
settings \cite{Cerezo2025}.  Our local results establish neither efficient
global search nor asymptotic barren-plateau suppression or noisy-hardware
advantage.  They identify a concrete structural link between the existence
of useful quantum circuits and the possibility of finding them through
training.

\section{Mixed discrete--continuous circuit space}
\label{sec:mixed}

We consider circuits built from CNOT gates and arbitrary single-qubit
rotations.  Grouping consecutive gates of the same type gives an
alternating sequence of local layers and CNOT blocks,
\begin{equation}
 U(\bm A,\bm\theta)
 =L_t(\theta_t)C(A_t)L_{t-1}(\theta_{t-1})\cdots C(A_1)L_0(\theta_0).
 \label{eq:alternating}
\end{equation}
This grouping preserves the original gate order; identity layers allow the
circuit to begin or end with either type of gate.
We call each grouped CNOT-only subcircuit a \emph{logical CNOT routing
block}; it need not be one parallel gate layer.
Thus $t$ counts routing blocks, not hardware depth.
Under unrestricted logical connectivity, there are $n(n-1)$ directed
CNOT gates: $\binom n2$ pairs with two control--target directions each.
For a block containing any number $m$ of gates, let $E_j$ be the
elementary binary matrix of its $j$th gate.  Then
\begin{equation}
 \mathrm{CNOT}_{c_m\to q_m}\cdots\mathrm{CNOT}_{c_1\to q_1}=C(A),
 \qquad A=E_m\cdots E_1\in\GL(n,\Ftwo).
 \label{eq:routing-word-encoding}
\end{equation}
Although $C(A)$ is a $2^n\times2^n$ permutation unitary, its action is
specified by the $n^2$ binary entries of $A$:
\begin{equation}
 C(A)\ket{x}=\ket{Ax},\qquad A\in\GL(n,\Ftwo),\quad x\in\Ftwo^n.
 \label{eq:cnot-linear-map}
\end{equation}
Conversely, every such map can be synthesized using CNOT gates
\cite{PatelMarkovHayes08}.

For the continuous part, let $\mathcal L_n=\SU(2)^n$ denote tuples of
$n$ single-qubit gates, represented on the register by their tensor
product.  Using the $ZXZ$ Euler convention, the $\ell$th local layer is
\begin{equation}
 L_\ell(\bm\theta_\ell)=\bigotimes_{q=1}^n
 R_z(\alpha_{\ell q})R_x(\beta_{\ell q})R_z(\gamma_{\ell q}),
 \qquad \ell=0,\ldots,t,
 \label{eq:local-layer}
\end{equation}
where $\bm\theta_\ell$ collects the $3n$ angles and
$R_a(\varphi)=\exp(-i\varphi\sigma_a/2)$ for $a=x,z$.
Thus $A_\ell$ specifies the routing operator $C(A_\ell)$, while
$\bm\theta_\ell$ specifies the local operator $L_\ell(\bm\theta_\ell)$.
A circuit with $t$ routing blocks therefore has the mixed encoding
\begin{equation}
 \cE_{n,t}
 =\bigl(\GL(n,\Ftwo)\bigr)^t\times\mathcal L_n^{\,t+1}.
 \label{eq:coordinate-free-space}
\end{equation}
Here the continuous factors are local gate tuples; $\bm\theta$ provides
their $3n(t+1)$ Euler-angle coordinates.  The tuple
$\bm A=(A_1,\ldots,A_t)$ specifies the discrete routing.
Unlike the one-to-one correspondence between $A$ and $C(A)$, Euler-angle
coordinates are not unique.  We identify unitaries
that differ only by a global phase.  The encoding is redundant: different
coordinates, including different routing tuples, can represent the same
physical unitary.

\begin{proposition}[Mixed representation of polynomial-size circuits]
\label{prop:mixed-coverage}
Consider a uniform unitary circuit family over a fixed, effectively
compilable, bounded-arity gate library, with polynomially many qubits and
gates in the input length.  To operator-norm accuracy $\epsilon$ modulo
global phase, it admits a mixed encoding whose block count and angle-bit
precision are polynomial in the input length and $\log(1/\epsilon)$.
Conversely, under all-to-all logical connectivity, $t$ routing
blocks and their local layers have an implementation using $O(tn^2)$ CNOT
gates and $O(n(t+1))$ single-qubit gates.
\end{proposition}

The proof and precision bookkeeping are given in Appendix~\ref{app:bqp}.
Logical gate count does not specify hardware depth.

Taking the union over polynomial bounds on $t$ gives precisely the
expressive power of polynomial-size circuits over CNOT and arbitrary
single-qubit gates, with the precision convention above.
Each routing block admits an $O(n^2)$-CNOT implementation, while the local
layers require $O(n(t+1))$ single-qubit gates.  Conversely, grouping any
polynomial-size circuit produces only polynomially many routing blocks.

Polynomial-size circuits therefore admit polynomial-length descriptions
without explicitly storing their exponentially large unitary matrices.
For $n$ qubits and $t$ routing blocks, the coordinates are
\begin{equation}
 \underbrace{\bigl((A_\ell)_{ij}\bigr)\in\{0,1\}^{tn^2}}_
 {tn^2\ \text{binary entries}}
 \qquad\text{and}\qquad
 \underbrace{\bm\theta\in[0,2\pi)^{3n(t+1)}}_
 {3n(t+1)\ \text{rotation angles}}.
 \label{eq:explicit-coordinate-values}
\end{equation}
Each $A_\ell$ must be invertible over $\Ftwo$.
Angles are taken modulo $2\pi$ because global phase is ignored.
With $b_\theta$ bits per angle, the description uses
$tn^2+3n(t+1)b_\theta$ bits, apart from bookkeeping overhead.

In particular, a CNOT-only block, even one of exponential length, still
implements a single binary
map $A$ recorded by $n^2$ bits.  Inserting only polynomially many
single-qubit gates leaves polynomially many such blocks and local controls.
Adding CNOTs within a block can change its map, but does not add continuous
parameters.  With interleaved local rotations, further simplification
depends on how these rotations interact with the CNOT blocks.

A related discrete limit is obtained when every local Euler angle is an
integer multiple of $\pi/2$: the circuit is then Clifford, generated up to
global phase by CNOT, Hadamard ($H$), and phase ($S$) gates.  Its net unitary
admits a polynomial-size Clifford implementation even if the original
gate sequence is exponentially long \cite{AaronsonGottesman2004}.

Continuous non-Clifford angles alone do not ensure computational advantage
either.  Local rotations can sometimes commute through CNOT gates and
then merge or cancel, reducing the effective circuit complexity.
Even when such cancellations are incomplete, some circuits retain
structures that permit efficient tensor-network simulation.
Examples include circuits with low contraction width or suitably bounded
entanglement throughout the evolution \cite{MarkovShi2008,Vidal2003}.
Later, we use this representation to discuss where quantum computational
advantage may arise, beyond these efficiently simulable regimes.

Many commonly used variational ans\"atze alternate trainable single-qubit
rotations with a prescribed entangling pattern
\cite{Kandala2017,Cerezo2021}.  In CNOT-based architectures, this pattern
may be a chain, ring, or brick-wall layout.
Within our description, fixing the connectivity, CNOT directions, and
gate order fixes the routing tuple $\bm A=\bm A_0$, leaving only the local
angles to be trained.  Joint architecture--parameter search and
fixed-routing training therefore optimize over different domains:
\begin{equation}
 \min_{\bm A,\bm\theta}\cL(U_{\bm A}(\bm\theta))
 \qquad\text{versus}\qquad
 \min_{\bm\theta}\cL(U_{\bm A_0}(\bm\theta)).
 \label{eq:mixed-search}
\end{equation}
A fixed routing may suffice for a particular task; the restriction alone
does not imply failure.  We next examine how this choice controls both the
correlations a circuit can generate and the directions its parameters can
access.

\section{How routing restricts reachable unitaries}
\label{sec:osr}

We now study what the routing choice allows a circuit to represent.
Hold the routing tuple $\bm A$ fixed and allow all local angles to vary.
The resulting unitaries form the reachable set
\begin{equation}
 \mathcal M_{\bm A}
 =\{[U(\bm A,\bm\theta)]:\bm\theta\ \text{varies}\}.
 \label{eq:routing-image}
\end{equation}
Changing $\bm A$ can change this set, even though the continuous controls
remain the same.  We first identify when two single-block routings give
the same reachable set.  We then give a computable sufficient condition
for their sets to be disjoint.  Finally, we bound what a fixed sequence
of routing blocks can represent.  These are statements about all allowed
angles; local training directions are considered in the next section.

\subsection{When different routings give the same reachable set}

For one routing block,
$\mathcal M_A=\{[L_1C(A)L_0]:L_0,L_1\in\mathcal L_n\}$.
Changing $A$ need not change this set: a change that can be absorbed into
the freely chosen local layers adds no new reachable unitaries.
For a single block, this observation gives an exact dichotomy.

\begin{proposition}[Single-block image dichotomy]
\label{prop:single-block-images}
For any number of qubits $n$ and any $A,B\in\GL(n,\Ftwo)$, the
$n$-qubit single-block images $\mathcal M_A$ and
$\mathcal M_B$ are either identical or disjoint.  They are identical
exactly when
\begin{equation}
 C(B)=e^{i\phi}L\,C(A)\,R
 \quad\text{for some }L,R\in\mathcal L_n,\ \phi\in\mathbb R.
 \label{eq:single-block-equivalence}
\end{equation}
Thus equality of reachable sets allows the routing operators themselves
to differ by local gates before and after the block.
\end{proposition}

The proof is given in Appendix~\ref{app:osr-proof}.
As a two-qubit example within an $n$-qubit register, take
$C(A)=\mathrm{CNOT}_{c\to q}$ and
$C(B)=\mathrm{CNOT}_{q\to c}$.  Hadamard gates give
\begin{equation}
 \mathrm{CNOT}_{q\to c}
 =(H_c\otimes H_q)\,\mathrm{CNOT}_{c\to q}\,(H_c\otimes H_q).
 \label{eq:hadamard-cnot-reversal}
\end{equation}
This realizes Eq.~\eqref{eq:single-block-equivalence} with
$L=R=H_c\otimes H_q$, using identities on all other qubits and absorbing
the phases of the $\SU(2)$ representatives.  Thus the routing operators
are different, but their reachable sets are equal.
The following equivalent circuits are read from left to right:
\begin{center}
\begin{tikzpicture}[x=0.85cm,y=0.7cm,
 gate/.style={draw,fill=white,minimum width=0.38cm,minimum height=0.38cm,inner sep=1pt}]
 \node[left] at (0,1) {$c$};
 \node[left] at (0,0) {$q$};
 \draw (0,1)--(4,1);
 \draw (0,0)--(4,0);
 \node[gate] at (1,1) {$H$};
 \node[gate] at (1,0) {$H$};
 \node[gate] at (3,1) {$H$};
 \node[gate] at (3,0) {$H$};
 \draw (2,0)--(2,1);
 \fill (2,1) circle (2pt);
 \draw (2,0) circle (0.12cm);
 \draw (1.86,0)--(2.14,0);
 \draw (2,-0.17)--(2,0.17);
 \node at (4.65,0.5) {$=$};
 \draw (5.3,1)--(7.7,1);
 \draw (5.3,0)--(7.7,0);
 \draw (6.5,0)--(6.5,1);
 \fill (6.5,0) circle (2pt);
 \draw (6.5,1) circle (0.12cm);
 \draw (6.36,1)--(6.64,1);
 \draw (6.5,0.83)--(6.5,1.17);
\end{tikzpicture}
\end{center}
Thus opposite CNOT directions have the same single-block reachable set.
For a general $n$-qubit CNOT block, a layer of Hadamard gates on all
qubits before and after the block gives
\begin{equation}
 H^{\otimes n}C(A)H^{\otimes n}=C(A^{-\mathsf T}),
 \qquad \mathcal M_A=\mathcal M_{A^{-\mathsf T}}.
 \label{eq:hadamard-inverse-transpose}
\end{equation}
For a CNOT word this reverses each control--target direction while
preserving the temporal gate order.  Here the local factors are
$L=R=H^{\otimes n}$, up to global phase.  The binary Fourier proof is
given in Appendix~\ref{app:osr-proof}.

\subsection{A sufficient condition for disjoint reachable sets}

The Hadamard examples show how two routings can represent the same set.
To certify the opposite case, we use operator Schmidt rank: if its value
differs for $C(A)$ and $C(B)$ across one fixed qubit cut, local layers
cannot remove that difference.
Whether a particular output state is entangled is unsuitable because it
depends on the input:
CNOT preserves $|00\rangle$ but maps $|+0\rangle$ to a Bell state.
For a labeled qubit cut $S\mid\bar S$, the operator Schmidt rank instead
counts the minimum number of product operators needed to represent $U$:
\begin{equation}
 \OSR_{S\mid\barS}(U)
 =\min\left\{r:U=\sum_{\alpha=1}^{r}L_\alpha\otimes R_\alpha\right\}.
 \label{eq:osr-def}
\end{equation}
Here the $n$ qubits are labeled by $[n]=\{1,\ldots,n\}$.
The set $S$ is any nonempty proper subset of $[n]$, and
$\bar S=[n]\setminus S$ contains the remaining qubits.
The factors $L_\alpha$ and $R_\alpha$ act on the $|S|$-qubit and
$|\bar S|$-qubit subsystems, respectively; each factor may act jointly
on all qubits on its side.  The cut need not separate consecutive sites.
For fixed routing $A$ and a fixed cut, varying the local layers leaves
the operator Schmidt rank equal to that of $C(A)$.
Different values for $C(A)$ and
$C(B)$ therefore certify that their images are disjoint.

The binary matrix $A$ introduced in Sec.~\ref{sec:mixed} also determines
the operator Schmidt rank of $C(A)$ across any specified cut.
We can compute this invariant directly from $A$, without constructing
the full unitary.

\begin{samepage}
Partition the qubit labels into $S$ and its complement $\bar S$, and
write the input and output strings as $x=(x_S,x_{\bar S})$ and
$y=(y_S,y_{\bar S})=Ax$.  Listing the coordinates in $S$ first on both
the input and output sides gives the block representation
\begin{equation}
 A=
 \begin{pmatrix}
 A_{S,S}&A_{S,\bar S}\\
 A_{\bar S,S}&A_{\bar S,\bar S}
 \end{pmatrix}
 =\begin{pmatrix}P&Q\\R&T\end{pmatrix}.
 \label{eq:A-block}
\end{equation}
Rows label output bits and columns label input bits.  Thus, over $\Ftwo$,
\begin{equation}
 y_S=Px_S+Qx_{\bar S},\qquad
 y_{\bar S}=Rx_S+Tx_{\bar S}.
 \label{eq:binary-cut-action}
\end{equation}
\end{samepage}
The square blocks $P$ and $T$ describe dependence within each side.
The off-diagonal blocks describe dependencies across the cut:
$Q$ from $\bar S$ to $S$, and $R$ in the opposite direction.
They have dimensions $|S|\times|\bar S|$ and $|\bar S|\times|S|$,
respectively.  Their ranks over $\Ftwo$ count the independent binary
dependencies in each direction.  We therefore define their sum, the
CNOT bi-cutrank, by
\begin{equation}
 \kappa_A(S)=\rank_{\Ftwo}(Q)+\rank_{\Ftwo}(R).
 \label{eq:kappa}
\end{equation}
This sum determines the operator Schmidt rank of the entire routing
block exactly.

\begin{theorem}[Exact CNOT cut formula]
\label{thm:osr-formula}
For every $A\in\GL(n,\Ftwo)$,
\begin{equation}
 \OSR_{S\mid\barS}(C(A))=2^{\kappa_A(S)}.
 \label{eq:osr-formula}
\end{equation}
\end{theorem}

Given $A$, this is an exact polynomial-time calculation for a specified
cut.  If $k=|S|$, binary Gaussian elimination on the two off-diagonal
blocks takes
\begin{equation}
 O\!\left(k(n-k)\min\{k,n-k\}\right)
 \label{eq:cut-rank-cost}
\end{equation}
bit operations: $O(n^3)$ for a balanced cut and $O(n)$ for a single-qubit
cut.  These are the ranks of $Q$ and $R$; the full matrix $A$ always
has rank $n$.  The proof and computational bookkeeping are given in
Appendix~\ref{app:osr-proof}; related operator Schmidt rank results
appear in Refs.~\cite{NielsenEtAl2003,ChenYu2016}. The appendix also
connects this formula to the graph-state cut-rank relation
\cite{Hein2004GraphStates} through the vectorized routing operator.

\begin{corollary}[Sufficient condition for disjoint images]
\label{cor:local-separation}
For single-block images,
\[
 \exists S:\ \kappa_A(S)\ne\kappa_B(S)
 \quad\Longrightarrow\quad
 \mathcal M_A\cap\mathcal M_B=\varnothing.
\]
More generally, $C(A)$ and $C(B)$ cannot be related by pre- and
post-multiplication with unitaries local across that cut.
\end{corollary}

This follows from invariance of operator Schmidt rank under local factors.
The labeled profile
\begin{equation}
 q(A)=\bigl(\kappa_A(S)\bigr)_{
 \varnothing\ne S\subsetneq[n],\,S\sim\barS}
 \label{eq:profile}
\end{equation}
collects these tests across all cuts.  Thus $q(A)\ne q(B)$ certifies
disjointness.  The condition is sufficient, not necessary:
$q(A)=q(B)$ alone does not decide whether the reachable sets are equal
or disjoint.
The polynomial cost above applies to one specified cut; explicitly
listing the full profile entails $2^{n-1}-1$ distinct cuts.

\subsection{Capacity bounds for several routing blocks}

With several routing blocks, intermediate local rotations can change
the final OSR even when the routing is fixed, and distinct reachable
sets can overlap without being equal.
For example, two identical CNOTs cancel when the intervening local
layer is the identity, but inserting $R_z(\pi/2)$ on the target between
them produces the entangling unitary $e^{-i\pi Z\otimes Z/4}$.
Thus, for $C(A)=\mathrm{CNOT}_{1\to2}$, the image of $(A,A)$ strictly
contains that of $(I,I)$, which consists only of local unitaries
(Appendix~\ref{app:osr-proof}).
Although the final OSR is no longer fixed, the routing matrices still
bound it from above, uniformly over all choices of local angles.

\begin{theorem}[Cumulative cut budget]
\label{thm:cumulative}
For the alternating circuit \eqref{eq:alternating}, let
\[
K_{\bm A}(S)=\sum_{\ell=1}^{t}\kappa_{A_\ell}(S).
\]
Then
\begin{equation}
 \OSR_{S\mid\barS}\!\left(U(\bm A,\bm\theta)\right)
 \le
 \min\!\left\{
 2^{K_{\bm A}(S)},
 4^{\min(|S|,|\barS|)}
 \right\}.
 \label{eq:cumulative}
\end{equation}
For a normalized pure product input across the cut, the output Schmidt rank obeys
\begin{equation}
 \operatorname{SR}_{S\mid\barS}
 \left(U\ket{\psi_S}\ket{\psi_{\barS}}\right)
 \le
 \min\!\left\{
 2^{K_{\bm A}(S)},
 2^{\min(|S|,|\barS|)}
 \right\}.
 \label{eq:state-cumulative}
\end{equation}
\end{theorem}

The proof uses submultiplicativity of OSR and is given in
Appendix~\ref{app:cumulative-proof}.  Routing cancellation and intermediate
angles can make the bound loose; different upper bounds do not by
themselves certify disjoint images.  A target whose OSR exceeds the bound,
however, lies outside the corresponding reachable set.
Multiblock OSR need not even be a power of two: an explicit OSR-three
example is given in Appendix~\ref{app:osr-three}.

\paragraph{Shallow circuits on a chain.}
The routing-block count $t$ describes the representation, whereas the
CNOT-layer depth $d$ describes a chosen implementation on the specified
connectivity; $d$ need not be the minimum achievable depth.
Routing-block count and CNOT-layer depth are related by polynomial
overhead, but are not interchangeable.  Each routing block admits a
polynomial-depth implementation on either all-to-all connectivity or a
nearest-neighbor chain with both CNOT directions available, while each
parallel CNOT layer can itself be treated as a routing block.
Their quantitative relation depends on connectivity, compilation, and
cancellations.
We count only CNOT layers in $d$, allowing arbitrary single-qubit gates
before, between, and after them.

Both all-to-all and nearest-neighbor CNOT circuits fit the same routing
representation $A\in\GL(n,\Ftwo)$: a CNOT-only block implements a product
$A=E_g\cdots E_1$ of elementary binary matrices.
All-to-all connectivity allows each factor to couple any pair of distinct
qubits, whereas a chain allows only neighboring pairs.
With unrestricted depth and both CNOT directions available, either
architecture can realize every $A$; the difference lies in the depth
required, not in the routing representation itself.
Although each gate on the chain acts only on neighboring qubits,
successive gates can propagate input dependencies along the chain.
We now examine how a finite depth budget limits cross-cut capacity.

The cut budget has a simple interpretation for shallow circuits.
Let $g_S$ be the number of CNOT gates crossing a chosen cut in a
specified implementation.  Each such gate can multiply the Schmidt
rank by at most two; gates confined to either side cannot increase it.
Thus the cross-cut connections limit what angle training can generate
(Fig.~\ref{fig:shallow-cut-budget}).

\begin{corollary}[Capacity of shallow one-dimensional circuits]
\label{cor:shallow-chain}
Consider an $n$-qubit nearest-neighbor chain with arbitrary single-qubit
gates and $d$ layers of pairwise disjoint CNOT gates.  For every cut
between consecutive sites, $g_S\le d$ and
\begin{equation}
 \OSR_{S\mid\barS}(U)\le 2^{g_S}\le 2^d.
 \label{eq:shallow-osr}
\end{equation}
From a pure product input, the output satisfies
\begin{equation}
 \operatorname{SR}_{S\mid\barS}(|\psi_{\rm out}\rangle)\le 2^d,
 \qquad E_S(|\psi_{\rm out}\rangle)\le d,
 \label{eq:shallow-state}
\end{equation}
where $E_S=-\Tr(\rho_S\log_2\rho_S)$ is the entanglement entropy in bits.
\end{corollary}

The proof is given in Appendix~\ref{app:cumulative-proof}.

These bounds also tell us how much depth a given routing requires.
A single-block unitary $U=L_1C(A)L_0$ has OSR $2^{\kappa_A(S)}$
across each cut, regardless of the endpoint local gates.
On a nearest-neighbor chain, a depth-$d$ implementation can provide
OSR at most $2^d$ across a consecutive-site cut.
To implement $U$ exactly, the available capacity must therefore meet
the required capacity at every such cut:
\begin{equation}
 d\ge \max_{1\le k<n}\kappa_A(\{1,\ldots,k\}).
 \label{eq:routing-depth-lower-bound}
\end{equation}
Here the qubits retain their fixed order $1,\ldots,n$ along the chain.
Thus one routing block ($t=1$) need not mean one execution layer:
its cross-cut ranks can force a large depth.  This is a necessary lower
bound, not a guarantee that the indicated depth suffices.

The same depth constraint limits state preparation from pure product
inputs.  If $d=O(\log n)$, the output Schmidt rank across every
consecutive-site cut is at most polynomial in $n$, and its entanglement
entropy is at most $O(\log n)$.
For example, define a target of $m$ Bell pairs across $S\mid\bar S$ by
\begin{equation}
 \ket{\Phi_m}=2^{-m/2}\sum_{z\in\{0,1\}^{m}}\ket{z}_S\ket{z}_{\bar S},
 \label{eq:bell-target}
\end{equation}
possibly tensored with fixed local states.  Its Schmidt rank is $2^m$, so
preparing it exactly requires $d\ge m$.  More generally:
\begin{lemma}[Rank--fidelity bound]
\label{lem:rank-fidelity}
If $\ket{\psi}$ has Schmidt rank at most $R$ across the target cut, then
\begin{equation}
 \abs{\braket{\Phi_m}{\psi}}^2\le \frac{R}{2^m}.
 \label{eq:rank-fidelity}
\end{equation}
\end{lemma}
The proof is given in Appendix~\ref{app:rank-fidelity-proof}.  Hence
insufficient depth also limits approximation: the squared fidelity obeys
$F\le\min\{1,2^{d-m}\}$.
This bottleneck depends on connectivity.  With all-to-all connections,
$m$ disjoint cross-cut CNOTs can run simultaneously, generating the
same $m$ Bell pairs in one CNOT layer from
$|+\rangle^{\otimes m}|0\rangle^{\otimes m}$.

\begin{figure}[!htbp]
\centering
\includegraphics[width=\linewidth]{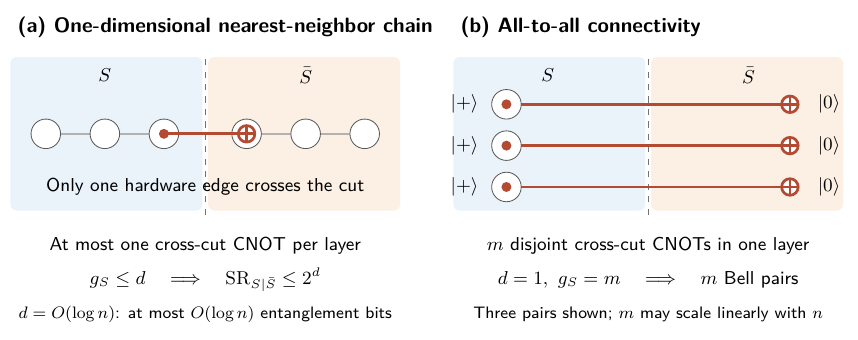}
\caption{\textbf{Depth limits cross-cut capacity through connectivity.}
Blue and orange regions specify a fixed qubit bipartition.  Red links show
CNOT gates, with a filled control and an open target marked by a plus.
\textbf{(a)} On a nearest-neighbor chain, only one hardware edge crosses
a consecutive-site cut.  At most one CNOT can use it in each parallel
layer, so depth $d$ permits at most $d$ cross-cut gates.  From a pure
product input, the output Schmidt rank is at most $2^d$ and the
entanglement entropy at most $d$ bits.  These are upper bounds, not
guarantees of saturation.  Gray links denote other available chain edges.
\textbf{(b)} With all-to-all connectivity, $m$ CNOTs on disjoint pairs can
cross the same cut simultaneously.  The indicated product input yields
$m$ Bell pairs in one CNOT layer.  Single-qubit gates are not included
in the CNOT-layer depth; preparing $|+\rangle$ from $|0\rangle$ adds
a local Hadamard layer.}
\label{fig:shallow-cut-budget}
\end{figure}

These bounds constrain capacity, but do not specify how the circuit
changes when its angles are varied.  Two routings can have identical
cut profiles yet turn local angle variations into different Pauli
directions (Appendix~\ref{app:profile-counterexample}).
Capacity alone therefore does not determine whether training can move
toward a target.  The next section examines which directions are
accessible by small angle changes at a given parameter point.

\section{Routing controls local tangent accessibility}
\label{sec:tangent}

The previous section considered the full set of unitaries reachable as the
angles vary.  We now fix the routing tuple $\bm A$ and write
$U_{\bm A}(\bm\theta)=U(\bm A,\bm\theta)$, where
$\bm\theta=(\theta_1,\ldots,\theta_p)$ collects the $p$ continuous gate
parameters.  Given an objective $\cL$, parameter training within this routing
sector seeks
\begin{equation}
 \min_{\bm\theta}\ \ell_{\bm A}(\bm\theta),
 \qquad
 \ell_{\bm A}(\bm\theta):=\cL\!\left(U_{\bm A}(\bm\theta)\right).
 \label{eq:fixed-routing-training}
\end{equation}
Thus $\bm\theta$ is the trained continuous variable, whereas $\bm A$ remains
fixed.  Near the current iterate $\bar{\bm\theta}$, a small parameter update
$\delta\bm\theta=(\delta\theta_1,\ldots,\delta\theta_p)$ changes the circuit to
first order as
\begin{equation}
 U_{\bm A}(\bar{\bm\theta}+\delta\bm\theta)
 \simeq U_{\bm A}(\bar{\bm\theta})
 +\sum_{j=1}^{p}\partial_jU_{\bm A}(\bar{\bm\theta})\,\delta\theta_j,
 \qquad
 \partial_j:=\frac{\partial}{\partial\theta_j}.
 \label{eq:circuit-first-order}
\end{equation}
Each derivative $\partial_jU_{\bm A}$ is therefore a circuit direction
available to training.  The central point is that these directions depend on
$\bm A$, even though training updates only $\bm\theta$.

\subsection{How routing propagates a local rotation}

The $j$th parameter $\theta_j$ controls a single-qubit rotation
$R_{P_j}(\theta_j)=e^{-i\theta_jP_j/2}$, where $P_j=X_q$ or $Z_q$ for some
qubit $q$ under the $ZXZ$ convention of Eq.~\eqref{eq:local-layer}.  Its
derivative is
\begin{equation}
 \frac{\partial R_{P_j}(\theta_j)}{\partial\theta_j}
 =-\frac{i}{2}P_jR_{P_j}(\theta_j).
 \label{eq:local-rotation-derivative}
\end{equation}
For parameter $\theta_j$, write the full circuit as
\begin{equation}
 U_{\bm A}
 =U_{{\rm after},j}R_{P_j}(\theta_j)U_{{\rm before},j}.
 \label{eq:circuit-split-at-parameter}
\end{equation}
Here $U_{{\rm before},j}$ acts first, before the parameter gate, and
$U_{{\rm after},j}$ acts later, between that gate and the output; matrix
products act from right to left.  Both parts depend on the routing $\bm A$
and on the other angles.  Differentiating only the middle gate gives
\begin{equation}
 \partial_jU_{\bm A}
 =-\frac{i}{2}U_{{\rm after},j}P_jR_{P_j}(\theta_j)
 U_{{\rm before},j}.
 \label{eq:full-circuit-parameter-derivative}
\end{equation}
Thus the derivative is the original circuit with $P_j$ inserted at the
parameter gate, together with the factor $-i/2$.  To follow this direction
forward to the circuit output, define
\begin{equation}
 K_j:=i(\partial_jU_{\bm A})U_{\bm A}^\dagger
 =\frac12U_{{\rm after},j}P_jU_{{\rm after},j}^\dagger.
 \label{eq:output-frame-generator}
\end{equation}
The gates before the parameter and the rotation itself cancel by unitarity.
For an independently parameterized Pauli rotation, $K_j$ is independent of
its own angle $\theta_j$ but still depends on the routing and on the other
angles contained in $U_{{\rm after},j}$.
This exact formula compresses every gate after the parameter into one
conjugation; it does not discard those gates.  Indeed, if
$U_{{\rm after},j}=V_m\cdots V_2V_1$, then
\begin{equation}
 K_j=\frac12V_m\cdots V_2V_1P_j
 V_1^\dagger V_2^\dagger\cdots V_m^\dagger,
 \label{eq:sequential-generator-propagation}
\end{equation}
For a fixed parameter $j$, this propagation can be read recursively as
\begin{equation}
 Q_0=P_j,\qquad Q_r=V_rQ_{r-1}V_r^\dagger\quad(r=1,\ldots,m),
 \qquad K_j=\frac12Q_m.
 \label{eq:generator-propagation-recursion}
\end{equation}
The next three steps describe this update when $V_r$ is one CNOT, a complete
CNOT routing block, or a single-qubit layer.  Changing the routing changes
some of these $V_r$ and can therefore change the final direction $K_j$.
Repeating the same calculation for $j=1,\ldots,p$ gives all directions
available from the circuit parameters.

\paragraph{One CNOT.}
The basic propagation step is already visible for a CNOT with control $c$
and target $t$.  Writing $C=\mathrm{CNOT}_{c\to t}$, conjugation gives
\begin{align}
 CX_cC^\dagger&=X_cX_t, & CX_tC^\dagger&=X_t, \nonumber\\
 CZ_cC^\dagger&=Z_c, & CZ_tC^\dagger&=Z_cZ_t.
 \label{eq:pauli-cnot}
\end{align}
Thus an $X$ direction on the control spreads forward to the target, while a
$Z$ direction on the target spreads backward to the control; $X_t$ and $Z_c$
remain unchanged.  Equivalently,
$R_P(\theta)=e^{-i\theta P/2}$ gives, for every angle (with the same notation
also used for a multiqubit Pauli generator),
\begin{align}
 C R_{Z_c}(\theta)&=R_{Z_c}(\theta)C,&
 C R_{X_t}(\theta)&=R_{X_t}(\theta)C,
 \label{eq:local-rotation-commutation}\\
 C R_{X_c}(\theta)&=R_{X_cX_t}(\theta)C,&
 C R_{Z_t}(\theta)&=R_{Z_cZ_t}(\theta)C.
 \label{eq:rotation-propagation}
\end{align}
Equations~\eqref{eq:local-rotation-commutation} leave the Pauli direction
local: $Z_c$ and $X_t$ pass through the CNOT unchanged.  In
Eq.~\eqref{eq:rotation-propagation}, by contrast, $X_c$ and $Z_t$ acquire a
Pauli factor on the other qubit.  Repeating this rule through several CNOTs
can spread one local direction over many qubits.  For example, $X_1$ becomes
$X_1X_2$ after $\mathrm{CNOT}_{1\to2}$; passing next through
$\mathrm{CNOT}_{2\to3}$ maps the $X_2$ factor to $X_2X_3$, giving
$X_1X_2X_3$.  Thus one scalar angle can control a multiqubit direction in the
full circuit.

\paragraph{A complete routing block.}
For a CNOT-only block, the same binary matrix $A$ used to measure capacity
also gives the exact propagation through the whole block.  For
$\bm x,\bm z\in\Ftwo^n$, write
$X^{\bm x}=\bigotimes_q X_q^{x_q}$ and
$Z^{\bm z}=\bigotimes_q Z_q^{z_q}$, where a zero bit denotes an identity
factor.  Then
\begin{equation}
 C(A)X^{\bm x}C(A)^\dagger=X^{A\bm x},\qquad
 C(A)Z^{\bm z}C(A)^\dagger=Z^{A^{-\mathsf T}\bm z}.
 \label{eq:routing-pauli-main}
\end{equation}
This is the CNOT-only specialization of the standard binary Clifford action
on Pauli labels \cite{DehaeneDeMoor2003}; the representation of a reversible
CNOT circuit by $A\in\GL(n,\Ftwo)$ is also central to linear reversible
circuit synthesis \cite{PatelMarkovHayes08}.  Thus $A$ determines which $X$
factors appear, while its inverse transpose does the same for $Z$ factors.
The binary updates can spread an operator to new sites or cancel factors
already present.  Here the same $A$ links this local propagation rule to the
cross-cut capacity studied in the previous section.

\paragraph{Interleaved local layers.}
A general $U_{{\rm after},j}$ can contain several routing blocks separated by
single-qubit layers.  A local layer mixes $X,Y,Z$ at the occupied sites but
does not introduce new sites.  For the $ZXZ$ convention, the complete
single-qubit propagation rules are
\begin{align}
 R_z(\phi)X R_z(\phi)^\dagger
 &=\cos\phi\,X+\sin\phi\,Y,&
 R_x(\phi)X R_x(\phi)^\dagger&=X,\nonumber\\
 R_z(\phi)Y R_z(\phi)^\dagger
 &=\cos\phi\,Y-\sin\phi\,X,&
 R_x(\phi)Y R_x(\phi)^\dagger
 &=\cos\phi\,Y+\sin\phi\,Z,\nonumber\\
 R_z(\phi)Z R_z(\phi)^\dagger
 &=Z,&
 R_x(\phi)Z R_x(\phi)^\dagger
 &=\cos\phi\,Z-\sin\phi\,Y.
 \label{eq:local-pauli-propagation}
\end{align}
For a multiqubit Pauli string, these rules apply independently at each
occupied site, so a local layer changes Pauli axes and coefficients without
expanding the support.
Each later routing block then propagates every resulting Pauli term by
Eq.~\eqref{eq:routing-pauli-main}.

\paragraph{Clifford limit.}
As an extreme case, suppose the subsequent local angles are integer multiples
of $\pi/2$, so that $U_{{\rm after},j}$ is Clifford.  Clifford conjugation
maps each Pauli string to one Pauli string, up to sign, and the routing blocks
have the same property.  The entire layer-by-layer propagation therefore
keeps a single string:
\begin{equation}
 K_j=\pm\frac12P_s.
 \label{eq:clifford-generator-propagation}
\end{equation}
Its Pauli label can be tracked by binary updates without expanding a sum.

\paragraph{General continuous angles.}
For arbitrary local angles, the sine--cosine mixtures in
Eq.~\eqref{eq:local-pauli-propagation} can create several Pauli terms.  The
output direction then has the general expansion
\begin{equation}
 K_j=\sum_s c_{j,s}P_s,
 \label{eq:generator-pauli-expansion}
\end{equation}
where $s$ labels the $n$-qubit Pauli strings
$P_s\in\{I,X,Y,Z\}^{\otimes n}$ and the real coefficients $c_{j,s}$ are
determined by gates in the forward causal cone of $P_j$.  CNOT routing can
enlarge the support of a local generator, and subsequent generic local
rotations can mix the Pauli axes at each occupied site.  A single $K_j$ can
therefore contain exponentially many Pauli components.  At the current
parameter point, however, an update $\delta\theta_j$ scales this entire
fixed combination, not its components independently: $K_j$ remains one
parameter-controlled direction.  Pauli-basis support and tangent-space
dimension therefore quantify different notions of complexity.

At a fixed parameter point, the real span of $K_1,\ldots,K_p$ is the
\emph{Pauli accessibility space} $\mathcal D_A$.  It records the directions
available to parameter training, whereas OSR records cross-cut capacity.  The
parameter count gives
$\dim\mathcal D_A\le p=3n(t+1)$, with strict inequality when some directions
are dependent.  This pointwise span is not the dynamical Lie algebra: no
closure under commutators is taken.  The next subsection rewrites the same
directions using $H_j$ and compares their span with the objective gradient.

\subsection{Ambient and accessible gradients}

For a fixed routing $\bm A$, ordinary VQC training uses the parameter
gradient of
\begin{equation}
 \cL(U)=\Tr\!\left[O\,U\rho U^\dagger\right].
 \label{eq:loss}
\end{equation}
Writing $U=U_{\bm A}(\bm\theta)$ and
$\ell_{\bm A}(\bm\theta)=\cL(U)$, its familiar components are
\begin{equation}
 g_j:=\partial_j\ell_{\bm A},\qquad
 g=\nabla_{\bm\theta}\ell_{\bm A}.
 \label{eq:ordinary-parameter-gradient}
\end{equation}
For the independent Pauli rotations used here, each $g_j$ can be estimated
by the standard parameter-shift rule
\cite{Mitarai2018,Schuld2019}; repeated shots are required to estimate the
two shifted expectation values.  Appendix~\ref{app:first-principles} derives
the rule under the present gate convention and states its measurement and
shared-parameter qualifications.

The routing dependence of this familiar gradient is exposed by the standard
unitary-gradient identity
\begin{equation}
 H_j:=iU^\dagger\partial_jU
 =\frac12U_{{\rm before},j}^\dagger P_jU_{{\rm before},j},\qquad
 G(U):=i[U^\dagger O U,\rho],\qquad
 g_j=\inner{H_j}{G(U)},
 \label{eq:routing-gradient-identity}
\end{equation}
where $\inner{X}{Y}=\operatorname{Re}\Tr(X^\dagger Y)$ is the unnormalized
Hilbert--Schmidt inner product.  Here $H_j$ is the physical direction supplied
by parameter $\theta_j$, while $G$ represents the objective gradient when
all unitary directions are allowed.  Under this metric and at fixed norm,
$G$ is steepest ascent and $-G$ is steepest descent.  These are standard
Riemannian-gradient facts \cite{Absil2008,Wiersema2023};
Appendix~\ref{app:first-principles}
derives Eq.~\eqref{eq:routing-gradient-identity}, including its sign and the
factor $1/2$.

The new routing implication is immediate.  Define
\begin{equation}
 \cT_U:=\Span_{\mathbb R}\{H_1,\ldots,H_p\},\qquad
 \mathcal D_A=U\cT_UU^\dagger.
 \label{eq:frame-equivalence}
\end{equation}
Parameter training sees $G$ only through its overlaps with the available
directions $H_j$.  Thus, even when $G\ne0$, every parameter gradient vanishes
if $G$ is orthogonal to $\cT_U$; near-orthogonality makes the physically
visible slope small.  Routing changes the propagated $H_j$ and hence changes
$\cT_U$, so another routing may align the accessible directions
with a part of $G$ that was previously hidden.

Formally, write $G=\Pi_{\cT_U}G+G_\perp$, where $G_\perp$ is orthogonal to
$\cT_U$.  Since every $H_j\in\cT_U$,
$\inner{H_j}{G_\perp}=0$, and therefore
\begin{equation}
 g_j=\inner{H_j}{G}=\inner{H_j}{\Pi_{\cT_U}G},
 \qquad
 g=0\ \Longleftrightarrow\ \Pi_{\cT_U}G=0.
 \label{eq:parameter-gradient-sees-accessible-part}
\end{equation}
The exact equivalence applies to a zero gradient.  For a small nonzero
gradient, a better-aligned routing may increase the visible slope.

When $g=0$, every accessible parameter direction is orthogonal to $G$, so
$\Pi_{\cT_U}G=0$.  For nonzero $g$, its norm also depends on how the
accessible directions are parameterized.  For example, duplicating a
direction, $H_2=H_1$, changes the parameter gradient from $c$ to $(c,c)$,
where $c=\inner{H_1}{G}$.  Its norm increases from $|c|$ to
$\sqrt{2}|c|$, although the accessible space is unchanged.  Such exact
redundancy occurs at the zero-angle point of the $ZXZ$ parametrization,
where the two $Z$ slots in the same single-qubit layer generate identical
operator tangents.  The numerical routing tests in Section~\ref{sec:repair}
use independently parameterized $X,Y,Z$ rotations and do not rely on this
$ZXZ$ coordinate redundancy.  To account for
these overlaps and scales, we introduce the Gram matrix
\begin{equation}
 M_{jk}=\inner{H_j}{H_k}.
 \label{eq:gram}
\end{equation}

\begin{proposition}[Coordinate-invariant projected gradient]
\label{prop:gram}
For any possibly redundant generator list,
\begin{equation}
 \norm{\Pi_{\cT_U}G}^2=g^{\mathsf T}M^+g,
 \label{eq:projected-gradient}
\end{equation}
where $M^+$ is the Moore--Penrose pseudoinverse.
\end{proposition}
\begin{proof}
Let $\mathcal Hc=\sum_jc_jH_j$, so $M=\mathcal H^*\mathcal H$ and
$g=\mathcal H^*G$. The orthogonal projector onto $\cT_U=\im\mathcal H$
is $\mathcal HM^+\mathcal H^*$, hence
$\norm{\Pi_{\cT_U}G}^2=\inner{G}{\mathcal HM^+\mathcal H^*G}
=g^{\mathsf T}M^+g$.
\end{proof}

This is the standard Gram formula for an orthogonal projection.  If $M$ is
singular, $M^+$ discards zero or redundant coordinate combinations.  Hence
$\sqrt{g^{\mathsf T}M^+g}$, unlike $\lVert g\rVert$, is the physical objective
slope visible to the circuit and is invariant under nonsingular
reparameterization.  The formula is related to natural-gradient geometry
\cite{Stokes2020,Meyer2021}, but here it serves as a diagnostic rather than a
prescribed optimizer.  Computing the full $M$ requires $O(p^2)$ overlaps,
whose individual evaluation can be costly; Clifford reference points provide
an important exception.  Ordinary training requires only the measurable
gradient $g$; the numerical sections instead use a task-effective
state-space Gram matrix.
To distinguish this operator-space geometry from the state-space geometry
used in the numerical analysis, we label the existing quantities
$G_{\rm op}\equiv G$, $M^{\rm op}\equiv M$, and
$\cT_U^{\rm op}\equiv\cT_U$, where $\mathrm{op}$ stands for operator.

The comparison relevant to routing is therefore
\begin{equation}
 \underbrace{\norm{\Pi_{\cT_U}G}}_{\text{routing-accessible gradient}}
 \quad\text{versus}\quad
 \underbrace{\norm{G}}_{\text{ambient gradient}}.
 \label{eq:visible-versus-ambient-gradient}
\end{equation}
If the first vanishes while the second does not, the objective has a descent
direction that the fixed routing does not expose.  We now formalize this
obstruction and give exact and perturbative evidence for it.

\subsection{Routing-induced gradient obstruction}
\label{sec:obstruction}

We now isolate the case in which the accessible projection vanishes
despite a nonzero ambient gradient.

\begin{definition}[Routing-constrained stationary point]
\label{def:stagnation}
A fixed-routing parameter point representing $U$ is routing-constrained
stationary when its differential image satisfies
\begin{equation}
 \Pi_{\cT_U}G=0
 \quad\text{and}\quad
 \norm{G}>0.
 \label{eq:constrained-stationary}
\end{equation}
\end{definition}

The routing has not destroyed $G$; it has hidden $G$ from all available
first-order parameter directions.  We call this a \emph{routing-induced
accessibility obstruction}.  It is a pointwise geometric statement, not by
itself a barren-plateau claim about gradient concentration or asymptotic
scaling, although the two mechanisms can coexist.

\subsubsection{Exact obstruction in a Bell-pair task}

Section~\ref{sec:osr} defined the Bell target $\ket{\Phi_m}$ and established
its capacity-limited fidelity through Lemma~\ref{lem:rank-fidelity}.  We now
ask what happens to the gradients at a smooth optimum within that restricted
routing sector.  For a prepared pure state $\ket{\psi}$, use infidelity
\begin{equation}
 \cL(\psi)=1-\abs{\braket{\Phi_m}{\psi}}^2.
 \label{eq:bell-loss}
\end{equation}

\begin{theorem}[Routing obstruction]
\label{thm:bell-floor}
Let the circuit act on a normalized pure product input across $S\mid\bar S$,
$\rho=\ket{\psi_0}\bra{\psi_0}$, with the loss in
Eq.~\eqref{eq:bell-loss} and
$F=|\langle\Phi_m|U|\psi_0\rangle|^2$.
Suppose the circuit has cumulative cut budget $K_{\bm A}(S)=k<m$.  Then
\begin{equation}
 \cL_{\min}\ge 1-2^{k-m}.
 \label{eq:bell-floor}
\end{equation}
The bound is tight for routing families that open $k$ independent Bell
channels.  At a smooth interior parameter-space local minimizer with
$0<F<1$, the accessible operator projection vanishes while the ambient
operator gradient is nonzero:
\begin{equation}
 \Pi_{\cT^{\rm op}_U}G_{\rm op}=0,\qquad \norm{G_{\rm op}}_{\rm HS}>0.
 \label{eq:bell-stagnation}
\end{equation}
\end{theorem}

\begin{proof}
Theorem~\ref{thm:cumulative} bounds the pure output Schmidt rank by $2^k$.
Lemma~\ref{lem:rank-fidelity} gives $F\le2^{k-m}$ and hence the loss floor.
Preparing $k$ independent Bell pairs with local Hadamards and cross-cut
CNOTs, and leaving the remaining pairs in $\ket{00}$, gives
$F=2^{k-m}$ and attains equality.  At an interior local minimum every
one-coordinate restriction has a two-sided minimum, so $g_j=0$ for every
parameter.  Equation~\eqref{eq:parameter-gradient-sees-accessible-part} then
gives the zero projection.
For the nonzero ambient claim, put
$Q=U^\dagger\ket{\Phi_m}\bra{\Phi_m}U$, a rank-one projector.
Since $\rho$ is also rank one and $F=\Tr(Q\rho)$,
one has $Q\rho Q=FQ$ and hence $\Tr(Q\rho Q\rho)=F^2$.  Therefore
\[
 \norm{[Q,\rho]}_{\rm HS}^2
 =2F-2F^2=2F(1-F)>0.
\]
The fidelity loss has $G_{\rm op}=-i[Q,\rho]$, proving the claim.
\end{proof}

The loss floor is the representational conclusion.  At the pure-state,
smooth interior minimizer specified in the theorem, it becomes a local
routing obstruction: the accessible gradient vanishes while the ambient
gradient remains nonzero. The routing dictionary makes the missing
directions explicit through their orthogonality to the propagated
generators. Section~\ref{sec:repair} uses this description to add useful
directions while preserving the current physical point.

\begin{figure}[!htbp]
\centering
\includegraphics[width=.96\linewidth]{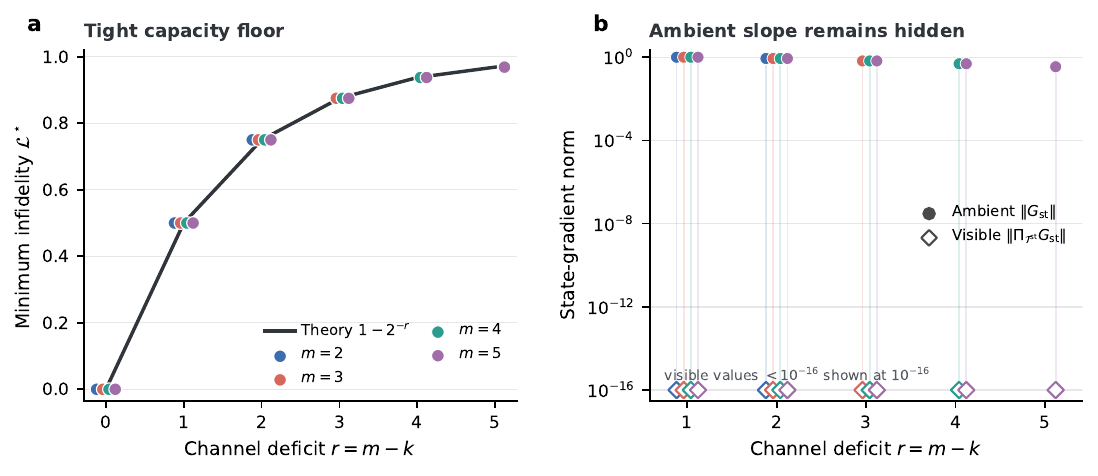}
\caption{\textbf{Bell-pair obstruction and constrained stationarity.}
Exact statevector evaluation of capacity-saturating constructions for
$m=2,3,4,5$, with one construction per $(m,k)$ point.
\textbf{(a)} When plotted against the channel deficit $r=m-k$, all
constructions attain the tight capacity floor
$\cL^\star=1-2^{-r}$; markers are offset slightly in the horizontal direction
to display coincident values for different $m$.  \textbf{(b)} At the corresponding constrained optima,
the ambient state-gradient norm remains nonzero while the visible norm
$\|\Pi_{\cT^{\rm st}}G_{\rm st}\|$ is at the numerical floor.  Both norms use
the pure-state horizontal metric; values below $10^{-16}$ are displayed at
$10^{-16}$.  There are no statistical error bars because the calculations
are deterministic.  Numerical details are given in
Appendix~\ref{app:numerics}.}
\label{fig:bell-floor}
\end{figure}

\subsubsection{General local stability of an exact obstruction}

The Bell construction motivates the question, but the argument is general.
At the exact obstruction, let $P$ be the orthogonal projector onto the
accessible space.  Then $PG$ is the part of the ambient gradient that the
parameters can see, and the obstruction is precisely $PG=0$ even though
$G\ne0$.  Now move to a nearby parameter point.  Write $G'$ for the new
ambient gradient and $P'$ for the new projector, so the new visible gradient
is $P'G'$.  Adding and subtracting $G$, and then using $PG=0$, gives
\begin{equation}
 \begin{aligned}
 P'G'
 &=P'(G'-G)+P'G\\
 &=\underbrace{P'(G'-G)}_{\text{the gradient changes}}
   +\underbrace{(P'-P)G}_{\text{the accessible space changes}},\\
 \norm{P'G'}&\le\norm{G'-G}+\norm{P'-P}_{\rm op}\norm{G}.
 \end{aligned}
 \label{eq:robust-projection}
\end{equation}
The bound says that a large visible gradient cannot appear immediately unless
either the ambient gradient changes substantially, measured by
$\norm{G'-G}$, or the accessible space rotates substantially, measured by
$\norm{P'-P}_{\rm op}$.  If both changes are small and the ambient gradient
remains nonzero, then only a small fraction of it is visible nearby.  Thus an
exact obstruction is normally surrounded by a region of near-stagnation.
The exception is a tangent-rank change, where a new accessible direction can
appear abruptly.  Appendix~\ref{app:near-stagnation} states the precise ratio
bound and a sufficient condition that rules out such an abrupt change.

\section{When routing exposes a usable gradient}
\label{sec:trainability}

Section~\ref{sec:tangent} showed that a small parameter gradient need not
reflect a small ambient objective gradient.  The ambient slope can remain
nonzero while the directions exposed by the current routing have little or
no overlap with it.  We therefore ask a local question: when does a routing
convert a nonzero ambient gradient into a usable parameter gradient?

Two ingredients control this conversion.  First, the accessible tangent
space must capture a sufficient fraction of the ambient gradient.  In a
fixed metric, write $G$ for the ambient gradient and $\mathcal T$ for the
accessible space.  Their alignment is quantified by the gradient coverage
\[
 a=\frac{\norm{\Pi_{\mathcal T}G}}{\norm{G}},\qquad G\ne0.
\]
Thus $a=0$ means that the ambient gradient is completely hidden from the
current parameter directions, whereas $a=1$ means that it lies entirely in
the accessible tangent space at the current point.  Intermediate values
$0<a<1$ quantify partial local alignment.

Second, even when the relevant direction is accessible, the parameter
coordinates must control it effectively.  After exact redundancies are
removed, nearly dependent effective directions can still make this control
weak.  This is the conditioning question, described below through the Gram
matrix of the tangent frame.  A usable parameter gradient therefore depends
on the ambient slope, directional coverage, and parameter conditioning.

To express these quantities for pure-state objectives, we now pass from the
operator tangents of Section~\ref{sec:tangent} to their action on the prepared
state.  This retains only the motion that is physically effective at the
current state.  For a fixed normalized pure input $\ket{\psi_0}$,
\[
 H_j=iU^\dagger\partial_jU,\qquad
 \ket{\psi}=U\ket{\psi_0},\qquad
 \partial_j\ket{\psi}=-iUH_j\ket{\psi_0}.
\]
Thus $H_j$ describes how parameter $\theta_j$ moves the unitary, while
$\partial_j\ket{\psi}$ describes how that motion acts on the current input.
A change of global phase has no physical effect.  Removing it leaves the
horizontal state tangent
\[
 v_j=(I-\ket{\psi}\bra{\psi})\partial_j\ket{\psi}.
\]
This is the physical state direction supplied by the same parameter;
an operator direction can have zero state tangent if it changes only phase.

A parameter update combines the effective state tangents $v_j$.  We collect
them as the columns of a linear map $V$, which sends a parameter-space
increment $c\in\mathbb R^p$ to the corresponding first-order physical state
displacement:
\[
 V=[v_1,\ldots,v_p],\qquad
 Vc=\sum_{j=1}^p c_jv_j.
\]
The accessible state-tangent space consists of all such combinations and is
therefore the image of this map:
\[
 \mathcal T^{\rm st}=\im V=\Span_{\mathbb R}\{v_1,\ldots,v_p\}.
\]
We use the real inner product
$\langle u,v\rangle_{\mathbb R}=\operatorname{Re}\braket{u}{v}$;
$V^*$ denotes the adjoint for this metric and real parameter coordinates.
For the expectation objective, the ambient state gradient and measured
parameter derivatives are
\[
 G_{\rm st}=2(O-\langle O\rangle)\ket{\psi},\qquad
 \langle O\rangle=\bra{\psi}O\ket{\psi},\qquad
 g_j=\operatorname{Re}\braket{v_j}{G_{\rm st}},\quad g=V^*G_{\rm st}.
\]
The operator and state metrics give different gradient representations,
$G_{\rm op}$ and $G_{\rm st}$, but the same loss derivatives $g_j$.
Below, $G$ and $\mathcal T$ denote the gradient and accessible space in the
chosen metric; for the state frame just defined they are $G_{\rm st}$ and
$\mathcal T^{\rm st}$.  The frame bound applies in either metric.

\subsection{From ambient slope to parameter gradient}
We now quantify how coverage and parameter conditioning jointly determine
the gradient.  For the tangent frame $V$, define the Gram matrix
\[
 M=V^*V,\qquad M^{\rm st}_{jk}=\operatorname{Re}\braket{v_j}{v_k}.
\]
Here $M=M^{\rm st}$ for the state frame.  Its smallest positive eigenvalue
$\lambda_{\min}^{+}(M)$ measures the weakest response among nonredundant
parameter combinations.  Exact redundancies give zero eigenvalues; nearly
dependent effective directions can give small positive ones.  The following
standard frame bound relates this spectrum to the visible gradient
\cite{Christensen2003}.

\begin{proposition}[Accessible-gradient frame bounds]
\label{prop:frame-bounds}
Let $V:\mathbb R^p\to\mathcal V$ synthesize a nonzero tangent frame,
$Vc=\sum_j c_j v_j$, in a real Hilbert space $\mathcal V$.  Set
$M=V^*V$, $g=V^*G$, and $\mathcal T=\im V$.  Then
\begin{equation}
 \lambda_{\min}^{+}(M)\norm{\Pi_{\mathcal T}G}^2
 \le \norm{g}_2^2
 \le \lambda_{\max}(M)\norm{\Pi_{\mathcal T}G}^2.
 \label{eq:frame-bounds}
\end{equation}
Here $\lambda_{\min}^{+}$ is the smallest strictly positive eigenvalue;
zero eigenvalues due to redundant or null columns are permitted.  If
$G\ne0$, define the coverage
$a=\norm{\Pi_{\mathcal T}G}/\norm{G}$.  A lower bound
$\lambda_{\min}^{+}(M)\ge\kappa>0$ therefore implies
\begin{equation}
 \norm{\nabla_\theta\cL}_2
 \ge\sqrt{\kappa}\,a\norm{G}.
 \label{eq:coverage-conditioning}
\end{equation}
\end{proposition}
\begin{proof}
Use the thin singular-value decomposition $V=W\Sigma Q^{\mathsf T}$.
The columns of $W$ form an orthonormal basis of $\mathcal T$, so
$\norm{\Pi_{\mathcal T}G}^2=\norm{W^*G}^2$ and
$\norm{g}^2=\norm{\Sigma W^*G}^2$.  Bounding the diagonal entries of
$\Sigma^2$ proves both inequalities.
\end{proof}

Equation~\eqref{eq:coverage-conditioning} bounds the parameter-gradient norm
from below by the ambient slope times coverage and $\sqrt{\kappa}$.
The projected norm is geometric, whereas $\norm{g}$ and the
spectrum of $M$ depend on parameter scaling.  To fix the scale of the
conditioning bounds, we use the physical angles in
$R_P(\theta)=e^{-i\theta P/2}$ throughout.  Nonsingular reparametrizations
preserve the bound when $g$ and $M$ are transformed accordingly.
If the frame has rank zero, both accessible and parameter
gradients vanish and $\lambda_{\min}^{+}$ is undefined.

This gradient lower bound yields a quantitative guarantee on the one-step
loss reduction under the following local assumptions.

\begin{corollary}[Conditional local descent]
\label{cor:local-descent}
Suppose the Euclidean gradient of the parameter loss is $L_s$-Lipschitz on
a neighborhood containing the segment from $\theta$ to
$\theta-h\nabla_\theta\cL$, where $0<h\le1/L_s$.  If at $\theta$
$\lambda_{\min}^{+}(M)\ge\kappa$, coverage is at least $a_0>0$, and
$\norm{G}\ge\gamma_0>0$, then
\begin{equation}
 \cL(\theta)-\cL(\theta-h\nabla_\theta\cL)
 \ge\frac{h}{2}\kappa a_0^2\gamma_0^2.
 \label{eq:one-step-descent}
\end{equation}
\end{corollary}
\begin{proof}
The descent lemma gives an improvement of at least
$h(1-L_sh/2)\norm{\nabla_\theta\cL}^2$; apply
Proposition~\ref{prop:frame-bounds}.
\end{proof}
The left-hand side is the current loss minus the loss after one update.
Thus ambient slope, coverage, and conditioning jointly give a guaranteed
minimum loss reduction for this sufficiently small gradient step.

\subsection{How routing controls local target overlap}
Why can the same target produce different local gradient responses under
different routings?  The preceding subsection characterizes the gradient
through aggregate coverage and conditioning; we now resolve the routing
dependence direction by direction.  Locally, this reduces to how well the tangent
frame of the training routing aligns with the displacement generated by the
target routing.  Compare a target circuit with routing tuple $\bm A$ and a training
circuit with routing tuple $\bm B$ at a common reference state.  For the CNOT circuits
considered below, both zero-angle circuits prepare $\ket{0^n}$.  Let $V_{\bm A}$
and $V_{\bm B}$ be their horizontal state synthesis maps at that reference point.
A fixed small target parameter vector $\theta^\star$ produces the first-order
target displacement $V_{\bm A}\theta^\star$, while the training parameters
$\theta$ produce the accessible displacement $V_{\bm B}\theta$:
\[
 V_{\bm A}\theta^\star=\sum_j\theta_j^\star v_j^{(\bm A)},
 \qquad V_{\bm B}\theta=\sum_j\theta_jv_j^{(\bm B)}.
\]
To second order, their infidelity is the squared distance between these
displacements,
\begin{equation}
 \cL_{\bm B}^{(2)}(\theta)
 =\norm{V_{\bm B}\theta-V_{\bm A}\theta^\star}_{\mathbb R}^2.
 \label{eq:cross-gram-quadratic}
\end{equation}
To see how training responds to this target, compare each training direction
with each target direction.  Their overlaps form the cross-Gram matrix
\[
 K_{\bm B,\bm A}=V_{\bm B}^*V_{\bm A},\qquad
 (K_{\bm B,\bm A})_{jk}=\operatorname{Re}\braket{v_j^{(\bm B)}}{v_k^{(\bm A)}}.
\]
Unlike the scalar coverage $a$ or a condition number, $K_{\bm B,\bm A}$ records
pairwise overlaps between two routing-dependent frames.  Differentiating
the quadratic loss gives
\[
 \nabla_\theta\cL_{\bm B}^{(2)}(\theta)
 =2V_{\bm B}^*(V_{\bm B}\theta-V_{\bm A}\theta^\star).
\]
At the reference point $\theta=0$, this becomes
\[
 \nabla_\theta\cL_{\bm B}^{(2)}(0)
 =-2V_{\bm B}^*V_{\bm A}\theta^\star
 =-2K_{\bm B,\bm A}\theta^\star.
\]
If the target displacement $V_{\bm A}\theta^\star$ is orthogonal to
$\im V_{\bm B}$, this leading gradient vanishes even when the target displacement
is nonzero.  Changing routing changes the training directions and hence
their overlap with the same target.

For fixed circuits, the exact initial infidelity gradient differs from this
expression by $O(\norm{\theta^\star}^2)$.  Hence, for
$\theta^\star=\epsilon u$ along a fixed direction $u$ with
$K_{\bm B,\bm A}u\neq0$, the exact initial gradient is nonzero for
sufficiently small nonzero $\epsilon$.
Appendix~\ref{app:frame-details} gives the expansion.

\subsection{An exact conditioning certificate at the CNOT origin}
Coverage determines whether useful directions are accessible; we now
certify whether their parameter frame is well conditioned.

At the CNOT zero-angle origin with $ZXZ$ local layers, each nonzero horizontal
tangent has a computational-basis binary label.  Distinct labels give
orthogonal directions; repeated labels give redundant copies.
Rank and conditioning follow directly from these labels.

Use $\ell$ local $R_zR_xR_z$ layers as in Eq.~\eqref{eq:local-layer},
interspersed with CNOT-only blocks, on input $\ket{0^n}$.
Let $A_j^{\rm suf}$ represent the CNOT suffix after layer $j$
(identity for an empty suffix).  With $e_q$ the $q$th binary unit vector,
$x_{jq}=A_j^{\rm suf}e_q$ labels the propagated single excitation.

\begin{theorem}[CNOT-origin state-tangent frame]
\label{thm:origin-frame}
At this origin both $Z$ slots per qubit have zero horizontal tangent.  The
$X$ slot in layer $j$ and qubit $q$ has tangent
\begin{equation}
 v_{jq}=-\frac{i}{2}\ket{x_{jq}},\qquad
 x_{jq}=A^{\rm suf}_j e_q\ne0.
 \label{eq:origin-label}
\end{equation}
If $m_x$ counts occurrences of the same nonzero binary label, then the
positive spectrum of $M^{\rm st}(0)$ is exactly
$\{m_x/4:x\text{ occurs}\}$.  Consequently,
\begin{equation}
 \rank M^{\rm st}(0)=\#\{x_{jq}\},\quad
 \lambda_{\min}^{+}(M^{\rm st}(0))\ge\tfrac14,\quad
 \lambda_{\max}(M^{\rm st}(0))\le\tfrac{\ell}{4}.
 \label{eq:origin-spectrum}
\end{equation}
Here $\#\{x_{jq}\}$ counts distinct labels.
After removing null parameter directions, the condition number is at most
$\ell$.  For any horizontal task gradient $G_{\rm st}$ at this point,
\begin{equation}
 \norm{\Pi_{\mathcal T^{\rm st}}G_{\rm st}}^2
 =\sum_{x:\,m_x>0}
   \bigl(\operatorname{Re}\langle-i x|G_{\rm st}\rangle\bigr)^2,
 \qquad |-i x\rangle:=-i\ket{x}.
 \label{eq:origin-coverage}
\end{equation}
Thus coverage of a nonzero relevant task component supplies the remaining
hypothesis for a local gradient lower bound with $\kappa=1/4$.
\end{theorem}

The bound $\lambda_{\min}^{+}\ge1/4$ ensures that the nonredundant
parameter combinations have no arbitrarily weak response at this origin.

\begin{proof}
An $X$ slot gives $-iC(A_j^{\rm suf})X_q\ket{0^n}/2$, yielding
\eqref{eq:origin-label}; a $Z$ slot gives only phase and horizontalizes to zero.
Basis orthogonality partitions the nonzero Gram columns into blocks
$\tfrac14\mathbf1_{m_x}\mathbf1_{m_x}^{\mathsf T}$, where
$\mathbf1_{m_x}\in\mathbb R^{m_x}$ is the all-ones vector.
Each block has the sole positive eigenvalue $m_x/4$.
Suffix invertibility gives distinct labels within each layer, hence
$1\le m_x\le\ell$ and the spectral bounds.
The distinct $-i\ket{x}$ form an orthonormal real basis of the tangent
space, giving \eqref{eq:origin-coverage}.
\end{proof}

For matched target and training routings,
\[
 \bm B=\bm A\quad\Longrightarrow\quad K_{\bm A,\bm A}=V_{\bm A}^*V_{\bm A}=M_{\bm A}^{\rm st}.
\]
Thus the target-response matrix of the preceding subsection is precisely
the conditioning matrix certified here.

Having separated coverage from conditioning, Section~\ref{sec:repair} asks
how to add task-aligned directions when coverage is insufficient, without
moving the current physical point.

\section{Same-point routing augmentation and numerical validation}
\label{sec:repair}

The preceding analysis leaves a practical question: how should one search
for a useful routing $\bm A$ together with the continuous parameters
$\theta$?  Fixed-routing VQC optimizes only $\theta$, yet
Sections~\ref{sec:osr}--\ref{sec:trainability} show that routing controls
both reachable structure and the tangent directions exposed to training.
A routing can leave the current tangent space with little access to a
nonzero ambient gradient.

Even before optimizing $\theta$, each CNOT routing block has
\[
 |\GL(n,\Ftwo)|=\prod_{k=0}^{n-1}(2^n-2^k)=\Theta(2^{n^2})
\]
possible transformations.  Choosing $t$ blocks independently gives
$|\GL(n,\Ftwo)|^t=2^{\Theta(tn^2)}$ routing tuples.
More generally, even the required block count $t$ may be unknown.
Yet increasing it need not discard the trained circuit: in
Eq.~\eqref{eq:alternating}, appending $A_{t+1}=I$ and
$\bm\theta_{t+1}=0$ adds only identities and preserves the original unitary.
The current solution can thus be embedded in a larger circuit family.
Compact representation alone does not make discovery easy: a generic
classical outer loop that efficiently found circuits for arbitrary
compactly representable tasks using task feedback and polynomial-size
quantum-circuit evaluations would provide an exceptionally powerful
quantum-algorithm discovery mechanism.  We do not generally expect
unrestricted discovery to be easy; additional task structure may make
efficient discovery possible.

We approach the joint search through incremental growth guided by local
geometry, rather than unrestricted global architecture search.
Suppose training has already produced a circuit $(\bm A,\theta)$, but its
tangent space exposes too little of a nonzero ambient gradient.
Can we retain the trained solution while adding directions that improve
this coverage?  This uses the local geometry of
Sections~\ref{sec:tangent}--\ref{sec:trainability} to guide the next
structural update, without prescribing the final architecture in advance.

A similar distinction arises in classical neural networks: increasing
depth or width can enlarge the model family, but choosing an architecture
and optimizing its continuous weights remain distinct problems.
Residual learning~\cite{He2016} motivates an incremental design principle:
retain the learned mapping while adding trainable corrections.
We adopt this principle in a geometric form: preserve the circuit already
learned, while changing the directions available for what it can learn next.
Specifically, inserting the routing-conditioned block
$W_B(\phi)=C(B)^\dagger L(\phi)C(B)$ with $L(0)=\id$ leaves the current
unitary, state, loss, ambient gradient, and existing parameter gradients
unchanged at initialization.  Identity initialization preserves the current
physical map; the derivatives with respect to $\phi$ depend on $B$ and supply
routing-dependent tangent directions.  The physical point is preserved,
while the accessible tangent space can expand.
The local selection question is therefore which
routing-conditioned identity augmentation exposes a useful descent
direction missing from the current tangent space.  The construction and
candidate scores below implement this step.

Identity-initialized blocks have been used to control effective depth at
initialization~\cite{Grant2019Init}, and layerwise learning grows circuit
depth during optimization~\cite{Skolik2021Layerwise}. Here identity
insertion is tied to a discrete routing coordinate, so different
routing-conditioned tangent spaces can be compared at exactly the same
trained physical point.

\subsection{Same-point routing intervention and local selection}
\label{sec:local-primitive}

\paragraph{Preserve the point, enlarge the tangent space.}
For a candidate routing $B$, place a trainable local layer between a routing
circuit and its inverse:
\begin{equation}
 W_B(\phi)=C(B)^\dagger L(\phi)C(B),
 \qquad L(0)=\id.
 \label{eq:gadget}
\end{equation}
Let $j$ denote a temporal insertion position, not a qubit index, and split
the current circuit there:
\[
 U=U_{>j}U_{\le j},\qquad
 U_{j,B}(\phi)=U_{>j}W_B(\phi)U_{\le j}.
\]
Only the new parameters $\phi$ are initialized to zero; the trained
parameters $\theta$ are retained. Since the inserted gadget is then exactly
the identity, $U_{j,B}(0)=U$.

\begin{proposition}[Same-point tangent augmentation]
\label{prop:gadget}
Let the new coordinates $\phi$ be independent of all old parameters.
At $\phi=0$, inserting $W_B(\phi)$ anywhere in a circuit leaves the unitary,
state, loss, old-parameter gradients, and ambient gradient unchanged.  Its
new tangent directions are the conjugated local generators
$C(B)^\dagger P C(B)$, transported to the insertion point by the surrounding
circuit.
\end{proposition}

The identity holds for every old parameter value, so the old derivatives,
generators $H_k$, and tangent space are retained. Write the new local coordinates as
$\boldsymbol\phi=(\phi_1,\ldots,\phi_m)$ and initialize all of them at
$\boldsymbol\phi=\mathbf0$, where $L(\mathbf0)=\id$.
For each coordinate $\phi_\mu$, define the local infinitesimal generator
\[
 \Lambda_\mu
 =i\left.\frac{\partial L(\boldsymbol\phi)}{\partial\phi_\mu}
 \right|_{\boldsymbol\phi=\mathbf0}.
\]
For a layer of zero-angle parameterized rotations, all undifferentiated
gate factors are identities at this point. Using $H=iU^\dagger\partial U$ gives
\[
 H^{\rm new}_{j,B,\mu}
 =U_{\le j}^\dagger C(B)^\dagger\Lambda_\mu C(B)U_{\le j}.
\]
Routing $B$ conjugates $\Lambda_\mu$ into a routing-dependent direction,
while the circuit prefix
$U_{\le j}$ transports that direction into the common body frame.
The same routing can therefore expose different directions at different
insertion positions. For $R_P(\phi_\mu)=e^{-i\phi_\mu P/2}$,
$\Lambda_\mu=P/2$.

Write the old accessible space as
$\cT=\Span_{\mathbb R}\{H_1,\ldots,H_p\}$. The new generators span
\[
 \cD_{j,B}=\Span_{\mathbb R}\{H^{\rm new}_{j,B,\mu}\}_{\mu},
\]
so the augmented space at the same point is
$\cT_{\rm aug}=\cT+\cD_{j,B}$. New parameters need not add independent
tangent directions: their generators may overlap with one another or with
the old space $\cT$. This distinction motivates the incremental score below.

\paragraph{Score alignment and incremental access.}
A candidate-only score asks how much of the ambient gradient the candidate
bundle captures by itself. An incremental score instead asks how much
additional gradient becomes accessible when that bundle is combined with
the current circuit. For a candidate $(j,B)$, define
\begin{align}
 s_{\rm cand}(j,B)&=\norm{\Pi_{\cD_{j,B}}G}^2,
 \label{eq:pauli-score}\\
 \Delta_{j,B}&=
 \norm{\Pi_{\cT+\cD_{j,B}}G}^2-\norm{\Pi_{\cT}G}^2.
 \label{eq:positioned-delta}
\end{align}
Since $\cT\subseteq\cT+\cD_{j,B}$, $\Delta_{j,B}\ge0$.
The gain is zero if the candidate only repeats old directions, or if its
genuinely new component is orthogonal to $G$.
The same incremental gain has an equivalent residualized form.
Let $P=\Pi_{\cT}$. Since
\[
 \cT+\cD_{j,B}=\cT\oplus(I-P)\cD_{j,B},
\]
the subtraction in Eq.~\eqref{eq:positioned-delta} removes the contribution
already accessible through $\cT$, leaving
\[
 \Delta_{j,B}=\norm{\Pi_{(I-P)\cD_{j,B}}(I-P)G}^2.
\]
Thus, residualization compares the genuinely new candidate directions
with the part of the ambient gradient not already accessible to the
current circuit.

Under compatible local geometry, candidate-only alignment already equals
the true incremental gain:
\begin{theorem}[Candidate-only--incremental equivalence under commuting projectors]
\label{thm:equivalence}
Let $P=\Pi_{\cT}$ and $Q=\Pi_{\cD_{j,B}}$.  If $PG=0$ and $PQ=QP$, then
\begin{equation}
 \Delta_{j,B}=\norm{QG}^2=s_{\rm cand}(j,B).
 \label{eq:equivalence}
\end{equation}
In particular, the equality holds when $\cD_{j,B}\perp\cT$.
\end{theorem}

\begin{proof}
For commuting orthogonal projectors,
$\Pi_{\cT+\cD_{j,B}}=P+Q-PQ$.  Since $PG=0$ and $PQG=QPG=0$,
the projection of $G$ onto the sum is $QG$, while the old projected norm is
zero.
\end{proof}

At a stationary old point, commuting projectors therefore make the two
scores interchangeable, including when both vanish. Oblique overlap can
make the candidate-only score differ from the true incremental gain,
motivating the residualized form of $\Delta_{j,B}$ above.

Local gradient-based pool selection underlies
ADAPT-VQE~\cite{Grimsley2019}; Geo-ADAPT uses the pool information metric
to select the largest absolute natural-gradient component
\cite{Sohail2026GeoADAPT}. Candidate-only scoring is closely related to
this local selection approach, measuring task-gradient alignment with
a candidate subspace. Here the pool is structured by a finite-field
routing coordinate: the same matrix determines both the propagated
tangent dictionary and the global cut capacity of
Section~\ref{sec:osr}. Incremental scoring measures the newly accessible
component after removing directions already available to the old ansatz.
The connection between these global and local roles of routing is the
organizing principle of the construction.

\paragraph{Probe candidate directions.}
Scoring does not require reconstructing the full ambient gradient $G$.
Keep the old parameters fixed and all new parameters at zero, then vary
one candidate parameter to probe its directional loss response:
\begin{equation}
 g_{j,B,\mu}
 =\inner{H^{\rm new}_{j,B,\mu}}{G}
 =\left.
 \frac{\partial}{\partial\phi_\mu}
 \cL\!\left(U_{j,B}(\phi)\right)
 \right|_{\phi=0}.
 \label{eq:candidate-directional-derivative}
\end{equation}
Equivalently, this is the limit of
$[\cL(U_{j,B}(\delta e_\mu))-\cL(U)]/\delta$ as $\delta\to0$.
Depending on access, the response can be evaluated analytically, obtained
by parameter shift, or approximated by finite differences; experimental
estimates also incur finite-shot uncertainty. Projection scores additionally
require candidate--candidate Gram overlaps, and incremental scores require
old--candidate overlaps.

For the state-space metric considered here, the candidate Gram can be
accessed through expectation values without full state reconstruction.
To make this connection explicit, \emph{old} refers to the trained parameters $\theta_k$, and
\emph{new} to the inserted parameters $\phi_\mu$. Both sets act on the
same output state $|\psi\rangle=U|\psi_0\rangle$ at
$\boldsymbol\phi=\mathbf0$.
In the state-space metric of Section~\ref{sec:trainability}, their
physical first-order displacements are
\[
 v_k^{\rm old}=(I-|\psi\rangle\langle\psi|)\partial_{\theta_k}|\psi\rangle,
 \qquad
 v_\mu^{\rm new}=(I-|\psi\rangle\langle\psi|)\partial_{\phi_\mu}|\psi\rangle.
\]
These horizontal tangent vectors describe first-order physical changes
of the same output state. Their real inner products measure how much the parameter
directions overlap. The three required blocks are
\[
 \begin{aligned}
 (M_{\rm old}^{\rm st})_{kl}
 &=\operatorname{Re}\langle v_k^{\rm old}|v_l^{\rm old}\rangle,\\
 (M_{\rm new}^{\rm st})_{\mu\nu}
 &=\operatorname{Re}\langle v_\mu^{\rm new}|v_\nu^{\rm new}\rangle,\\
 (M_{\rm cross}^{\rm st})_{k\mu}
 &=\operatorname{Re}\langle v_k^{\rm old}|v_\mu^{\rm new}\rangle.
 \end{aligned}
\]
Candidate-only scoring uses $M_{\rm new}^{\rm st}$; incremental scoring
also needs the old block and the cross-block, which records overlap with
directions already available. Together they form
\[
 M_{\rm aug}^{\rm st}
 =\begin{pmatrix}
 M_{\rm old}^{\rm st}&M_{\rm cross}^{\rm st}\\
 (M_{\rm cross}^{\rm st})^{\mathsf T}&M_{\rm new}^{\rm st}
 \end{pmatrix}.
\]

The new--new block simplifies because both derivatives pass through the
same circuit suffix. At a fixed candidate $(j,B)$,
write the state just before insertion as $|\chi_j\rangle=U_{\le j}|\psi_0\rangle$ and
$Q_\mu=C(B)^\dagger\Lambda_\mu C(B)$, with expectations taken in
$|\chi_j\rangle$. At $\boldsymbol\phi=\mathbf0$,
\[
 \partial_{\phi_\mu}|\psi\rangle=-iU_{>j}Q_\mu|\chi_j\rangle,
 \qquad
 v_\mu^{\rm new}
 =-iU_{>j}(Q_\mu-\langle Q_\mu\rangle)|\chi_j\rangle.
\]
The second expression removes the global-phase component. The common
$U_{>j}$ cancels in the overlap, giving the candidate Gram
\[
 (M_{\rm new}^{\rm st})_{\mu\nu}
 =\operatorname{Re}\langle v_\mu^{\rm new}|v_\nu^{\rm new}\rangle
 =\operatorname{Re}\langle Q_\mu Q_\nu\rangle
  -\langle Q_\mu\rangle\langle Q_\nu\rangle.
\]
For Pauli rotations, $\Lambda_\mu=P_\mu/2$. Defining
$P'_\mu=C(B)^\dagger P_\mu C(B)$ gives $Q_\mu=P'_\mu/2$, and hence
\[
 (M_{\rm new}^{\rm st})_{\mu\nu}
 =\frac14\left[
 \operatorname{Re}\langle P'_\mu P'_\nu\rangle
 -\langle P'_\mu\rangle\langle P'_\nu\rangle\right].
\]
Thus the candidate Gram can be obtained from Pauli expectation
measurements, rather than a full state reconstruction
\cite{Stokes2020}. Collecting the new-parameter loss derivatives into
$g_{\rm new}$ gives the candidate-only score
\[
 s_{\rm cand}^{\rm st}(j,B)
 =g_{\rm new}^{\mathsf T}(M_{\rm new}^{\rm st})^+g_{\rm new}.
\]
Incremental scoring additionally requires the old--new cross-block in
$M_{\rm aug}^{\rm st}$, which need not admit the same Pauli simplification.
Together with the loss derivatives, the joint and old Gram matrices
determine the corresponding accessible gradient norms; subtracting the
old squared norm from the augmented squared norm gives $\Delta_{j,B}^{\rm st}$.
The insertion itself adds a local layer and two synthesized CNOT networks
while retaining the pre-existing routing.

For a fixed candidate with $r_B=O(n)$ new parameters, binary propagation
and construction of its insertion-frame Pauli dictionary take polynomial
time in $n$. Candidate-only scoring requires $r_B$ loss derivatives and
$O(r_B^2)$ Gram entries, followed by $O(r_B^3)$ dense linear algebra.
With $p$ old parameters, incremental scoring additionally uses $O(pr_B)$
cross-Gram entries and the reusable old gradient and $O(p^2)$ old Gram
data; a direct augmented-matrix calculation costs $O((p+r_B)^3)$.
Thus, conditional on obtaining these quantities at the required
precision, classical diagnostic processing for $K$ candidates costs
$O(K\,\mathrm{poly}(n,t))$, since $p=3n(t+1)$; this is
$O(K\,\mathrm{poly}(n))$ for polynomial block count $t$.
Single-candidate processing does not make exhaustive routing search
polynomial: $|\GL(n,\Ftwo)|=2^{\Theta(n^2)}$.
Hardware sampling costs depend on the measurement strategy, required
precision, and Gram conditioning. The experiments below evaluate these
quantities directly from statevectors.

\begin{center}
\fbox{\begin{minipage}{0.92\linewidth}
\textbf{One local intervention round.}
\begin{enumerate}[leftmargin=1.5em,itemsep=1pt,topsep=3pt]
\item Optimize the current ansatz to the intervention point.
\item Generate a finite library $\mathcal K=\{(j,B)\}$.
\item Probe directional responses; compute or estimate the required Gram overlaps.
\item Evaluate candidate-only or incremental scores using Gram pseudoinverses
or equivalent factorizations.
\item Select a candidate under the gate, depth, and measurement budget.
\item Insert it with new parameters at zero, retaining the old parameters.
\item Resume continuous optimization.
\end{enumerate}
\end{minipage}}
\end{center}

The intervention is designed primarily to improve gradient coverage.
Since coverage and Gram conditioning are distinct geometric properties,
a candidate with positive $\Delta_{j,B}$ need not improve conditioning.

The procedure is naturally iterative: after retraining, the local geometry can
be reassessed and another candidate inserted at the new optimization point.

\subsection{Controlled accessibility recovery}
\label{sec:controlled}

We first isolate one intervention round and test whether same-point
tangent expansion translates into subsequent optimization gain.

Tests use $(n,t)=(4,2),(6,3),(8,4)$ with five seeds each, giving 15
instances. Each has 12 equal-cost nonidentity candidates plus an identity
control, giving 195 insertions in total, followed by 50 training steps each.
Numerical local layers use independent rotations in time order
$X\to Y\to Z$ (matrix product $R_zR_yR_x$), extending beyond the analytic
$ZXZ$ certificate. Scores use exact statevectors and the horizontal metric
of Section~\ref{sec:trainability}, denoted by $\mathrm{st}$; all gadgets
are inserted at the input, and this fixed position is suppressed in score
notation. Appendix~\ref{app:statistics} defines the controls and collects
the protocol and statistical checks.

The controlled Bell-based test compares routing candidates at the
same trained point: identity initialization preserves the state, loss,
old gradients, and ambient gradient, while the candidates supply different
tangent directions. These quantities remain unchanged to machine zero
across all 195 insertions, isolating tangent accessibility in the
initial comparison.

The target contains Bell pairs across a chosen cut, dressed by
non-Clifford local layers and within-side routing, whereas the old ansatz
has no cross-cut CNOT. At its optimized points,
$\norm{G_{\rm st}}\in[0.484,0.866]$ but the visible-to-ambient norm ratio
lies in $[1.15\times10^{-12},9.24\times10^{-5}]$.
The obstruction is therefore lack of accessibility to an existing
descent direction, rather than absence of ambient descent.
Figure~\ref{fig:multiblock}(a) shows how insertion can expose this
gradient without changing the physical point.

The incremental score predicts one- and five-step descent with median
within-instance Spearman correlations $0.916$ and $0.902$; the corresponding
candidate-only values are $0.921$ and $0.902$. At 50 steps the medians
decrease to $0.534$ and $0.503$, respectively
(Appendix~\ref{app:statistics}). This horizon dependence is consistent
with a same-point local score: as training moves away from the scored
point, the initial geometry becomes less predictive of accumulated
descent. The representative training
trajectories in Fig.~\ref{fig:multiblock}(b) and the per-instance five-step
rank correlations in panel (c) connect initial geometry to realized
loss reduction. In panel (b), the target has four equal Schmidt weights
across the chosen cut. The gadget is inserted at the input, where
$C(B)\ket{0^n}=\ket{0^n}$; each selected nonidentity routing has operator
Schmidt rank two, and the remaining circuit acts within the two sides.
The resulting states therefore have Schmidt rank at most two, giving
fidelity at most $1/2$ by Lemma~\ref{lem:rank-fidelity} and explaining
the approach to infidelity $1/2$.
At 50 steps, incremental selection exceeds the fixed
random candidate in 13 of 15 instances; panel (d) compares selectors by
their frequency of attaining at least $90\%$ of the retrospective
finite-pool oracle gain.

Together, these results support the intended mechanism: routing candidates
that expose more of the previously inaccessible ambient gradient at the
same physical point tend to produce stronger descent after optimization
resumes. Candidate-only and incremental selection perform comparably in
this controlled test.

\begin{figure}[!htbp]
\centering
\includegraphics[width=.98\linewidth]{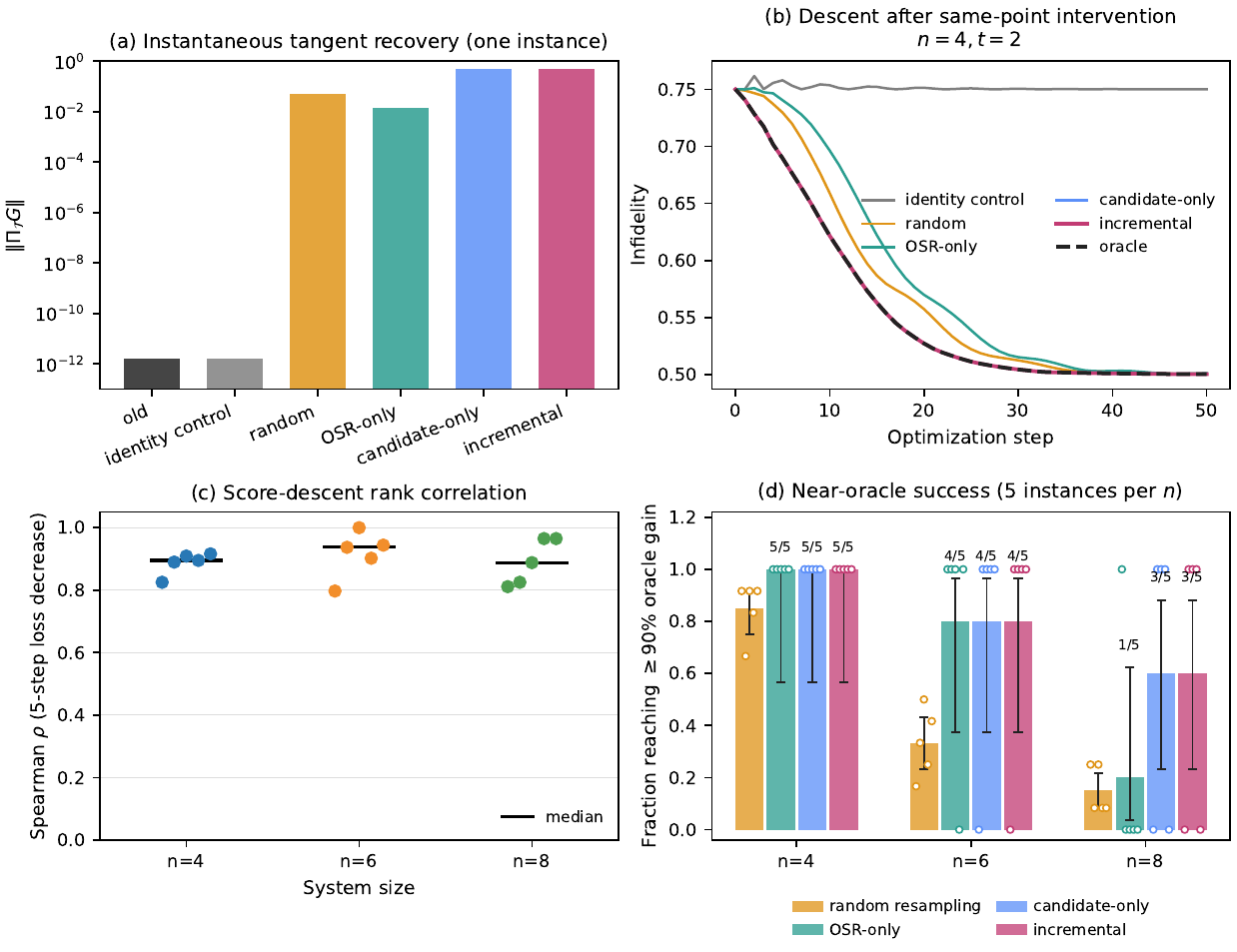}
\caption{\textbf{Same-point augmentation restores access and geometric scores
predict early descent.}
The controlled test compares 15 instances and 12 nonidentity candidates per
instance, each adding $3n$ parameters and eight CNOTs.
(a) Instantaneous accessible-gradient recovery at the same physical point
for one instance. All insertions start from the same state, loss, and
ambient gradient; values below $10^{-13}$ are shown at this display floor.
(b) Descent after same-point intervention, with infidelity
$1-|\langle\psi_{\rm target}|\psi\rangle|^2$ (lower is better).
For this representative instance (seed 20264923), candidate-only,
incremental, and retrospective oracle selection choose the same candidate,
so their trajectories coincide. The selected nonidentity circuits have
the structural infidelity lower bound $1/2$ described in the text.
(c) Within-instance Spearman correlations between incremental gain and
five-step loss decrease, each computed over 12 candidates. Each dot is
one instance; black bars mark medians over five instances per size.
Appendix~\ref{app:statistics} retains the 20-step candidate-rank scatter.
(d) Selector success fractions for attaining at least $90\%$ of the
retrospective finite-pool oracle's 50-step gain, with five instances per size.
Random bars show mean within-pool success probabilities.
Intervals are Wilson intervals for the three deterministic selectors and
instance-bootstrap intervals for random selection; the random baseline
uses 1000 saved-pool resamples per instance.}
\label{fig:multiblock}
\end{figure}

\subsection{Natural-task validation on a connected Heisenberg chain}
\label{sec:xxz}

The controlled Bell test isolates the recovery mechanism. We next test
whether geometry-guided selection remains useful for a natural many-body
objective when the baseline entangling graph is already connected.
Here augmentation changes the available directions within a connected
architecture, rather than simply connecting previously disconnected
subsystems. We use the open antiferromagnetic Heisenberg (XXX) chain,
\begin{equation}
 H_n=\sum_{q=0}^{n-2}
 \left(X_qX_{q+1}+Y_qY_{q+1}+Z_qZ_{q+1}\right)
 \label{eq:xxz}
\end{equation}
for $n=8,10,12$, with five fixed seeds per size. For a fixed routing
architecture $\bm A$, the variational task is to minimize
\[
 E(\bm A,\bm\theta)
 =\bra{\psi(\bm A,\bm\theta)}H_n\ket{\psi(\bm A,\bm\theta)},
 \qquad
 \ket{\psi(\bm A,\bm\theta)}=U(\bm A,\bm\theta)\ket{\psi_0},
\]
over the continuous parameters $\bm\theta$, toward the ground-state
energy $E_0$; Fig.~\ref{fig:xxz}(a) reports the energy error $E-E_0$.
The baseline has two
alternating nearest-neighbor brick-wall CNOT blocks whose union covers
every edge of the chain.

The comparison isolates which routing augmentation is selected.
At each independently selected stalled checkpoint, all 12 candidates
share the same insertion position between the blocks, add the same
$3n$ angles and $2(n-1)$ CNOTs, and receive the same 50-step Adam budget.
At insertion, identity initialization preserves the state, energy, and old
gradients, so all candidates are compared from the same physical point.
Candidate construction and scoring do not
use the exact ground state.

Figure~\ref{fig:xxz}(a) illustrates the contrast between continued
training of the old ansatz and descent after augmentation.
Random augmentation recovers descent in all 15 instances, as do both
geometric selectors.
Panel (b) addresses the stronger question: how effectively does each
selector use the same augmentation budget? Mean 50-step energy decreases
are $0.495$ for fixed random selection, $0.925$ for candidate-only,
$0.916$ for incremental, and $1.015$ for the retrospective finite-pool
oracle. Each geometric selector beats the fixed random candidate in
13/15 paired instances. Thus augmentation itself helps, but geometric
selection makes substantially better use of the same added resources.

Panels (c) and (d) provide complementary diagnostics of this selection.
Whereas the one- and five-step correlations in
Section~\ref{sec:controlled} test near-term descent, panel (c) tests whether
the local geometric signal remains informative over the full 50-step
post-intervention training budget. The within-pool rank associations vary
across instances.
Panel (d) compares the two geometric rankings directly. Incremental
scoring discounts directions already represented in the old tangent
space, changing the selected candidate in 4/15 pools. It
measures the additional gradient access beyond the old tangent space,
while candidate-only scoring requires only the new--new Gram block and
performs comparably in this benchmark.

On this connected many-body task, the benefit of architecture growth is
therefore not determined by the augmentation budget alone: local geometry
helps identify which equal-cost directions are useful at the current
optimization point.
Appendix~\ref{sec:state-geometry} collects the protocol,
paired statistics, and subspace diagnostics;
Appendix~\ref{app:tfim-control} gives the complementary TFIM selector
control, where the two geometric selectors nearly coincide.

\begin{figure}[!htbp]
\centering
\includegraphics[width=.98\linewidth]{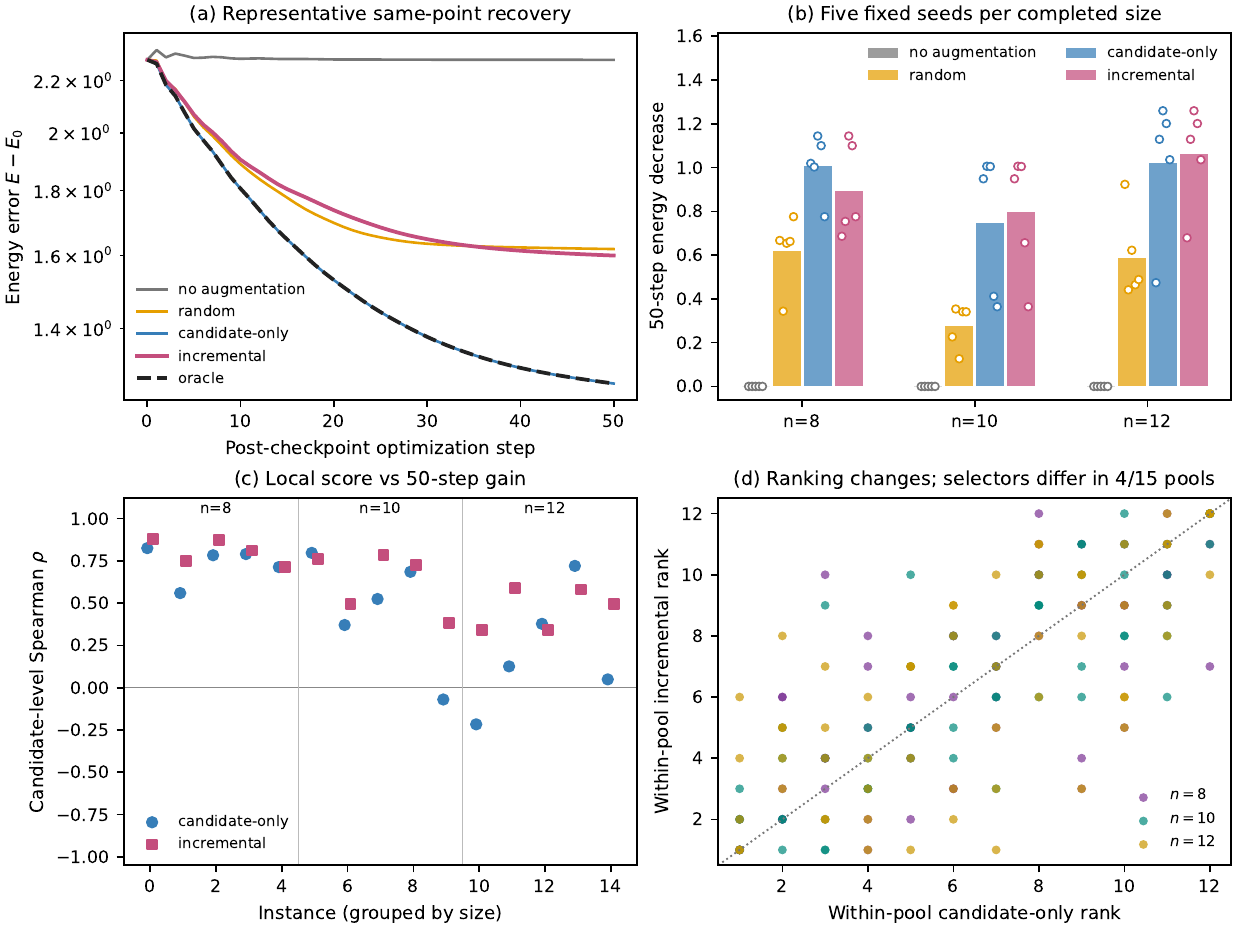}
\caption{\textbf{Geometric selection improves the use of equal-cost
augmentation on a connected Heisenberg ansatz.}
(a) Representative same-point continuation ($n=8$, seed 20268804).
(b) Fifty-step gains for five seeds at each $n=8,10,12$, with 12 equal-cost
candidates per checkpoint; bars show means and points show instances.
The random selector uses one fixed draw per instance.
(c) Within-instance score correlations with realized 50-step gain.
(d) Candidate-only and incremental scores can rank the same pool
differently, selecting different candidates in 4/15 instances.
The oracle is the retrospective best candidate.}
\label{fig:xxz}
\end{figure}

\subsection{Routing-dependent near-stagnation beyond the controlled construction}
\label{sec:genericity}

Finally, we test whether routing-dependent near-stagnation persists beyond
the engineered Bell construction. We use the open one-dimensional
transverse-field Ising model (TFIM),
\[
 H_{\rm TFIM}=-J\sum_{q=0}^{n-2}Z_qZ_{q+1}
              -h\sum_{q=0}^{n-1}X_q,
 \qquad J=h=1.
\]
For each fixed routing architecture $\bm A$, we minimize
\[
 E(\bm A,\bm\theta)
 =\bra{\psi(\bm A,\bm\theta)}H_{\rm TFIM}\ket{\psi(\bm A,\bm\theta)},
 \qquad
 \ket{\psi(\bm A,\bm\theta)}=U(\bm A,\bm\theta)\ket{0^n},
\]
over the continuous parameters $\bm\theta$, toward the ground-state
energy $E_0$. Whereas Section~\ref{sec:xxz} compares augmentation
selectors, this ensemble tests the incidence of near-stagnation across
routing families and initializations. For each $n=4,6,8$, we use
20 distinct routing architectures, four from each of five routing families.
Each architecture consists of two routing blocks ($t=2$) and is trained
from three initializations, giving 180 optimization instances in total.
These comprise open-chain line/brickwork routings, periodic ring routings,
fixed-budget random routings, repeated dimers with the same qubit pairing
and CNOT orientation in both blocks, and matched-cut-budget controls.
The last family matches its paired random routing's capacity budget
summed over all contiguous cuts,
$B_{\rm cut}(\bm A)=\sum_{j=1}^{n-1}K_{\bm A}(S_j)$, where
$S_j=\{0,\ldots,j-1\}$ and $K_{\bm A}(S)$ is the cumulative per-cut
budget of Theorem~\ref{thm:cumulative}. This control fixes the scalar
sum while allowing the detailed cut profile and routing geometry to differ.

Each run uses 120 Adam steps followed by at most 300 L-BFGS-B
iterations. At the returned terminal point, we evaluate the local
geometry and apply the near-stagnation criterion below.
Figure~\ref{fig:genericity} denotes the state-space gradient coverage
of Section~\ref{sec:trainability} by
$\chi_{\rm ratio}=a=\norm{\Pi_{\cT^{\rm st}}G_{\rm st}}/\norm{G_{\rm st}}$.
The preregistered primary event
combines a small parameter gradient with a nonzero ambient state gradient,
small visible coverage, and finite energy error; precise routing
constructions and numerical thresholds are given in
Appendix~\ref{app:genericity}.

The event occurs in 85/180 instances and for at least one initialization
of 39/60 routings. Its incidence depends strongly on routing: repeated
dimers give 36/36 events, whereas line and ring routings each give 0/12 at
$n=8$. At this size, none of the line or ring endpoints meets the
parameter-gradient threshold under this protocol, even though low coverage
remains common. Figure~\ref{fig:genericity} therefore shows that near-stagnation
extends beyond an isolated construction and is routing dependent in the
tested ensembles.

\begin{figure}[!htbp]
\centering
\includegraphics[width=\linewidth]{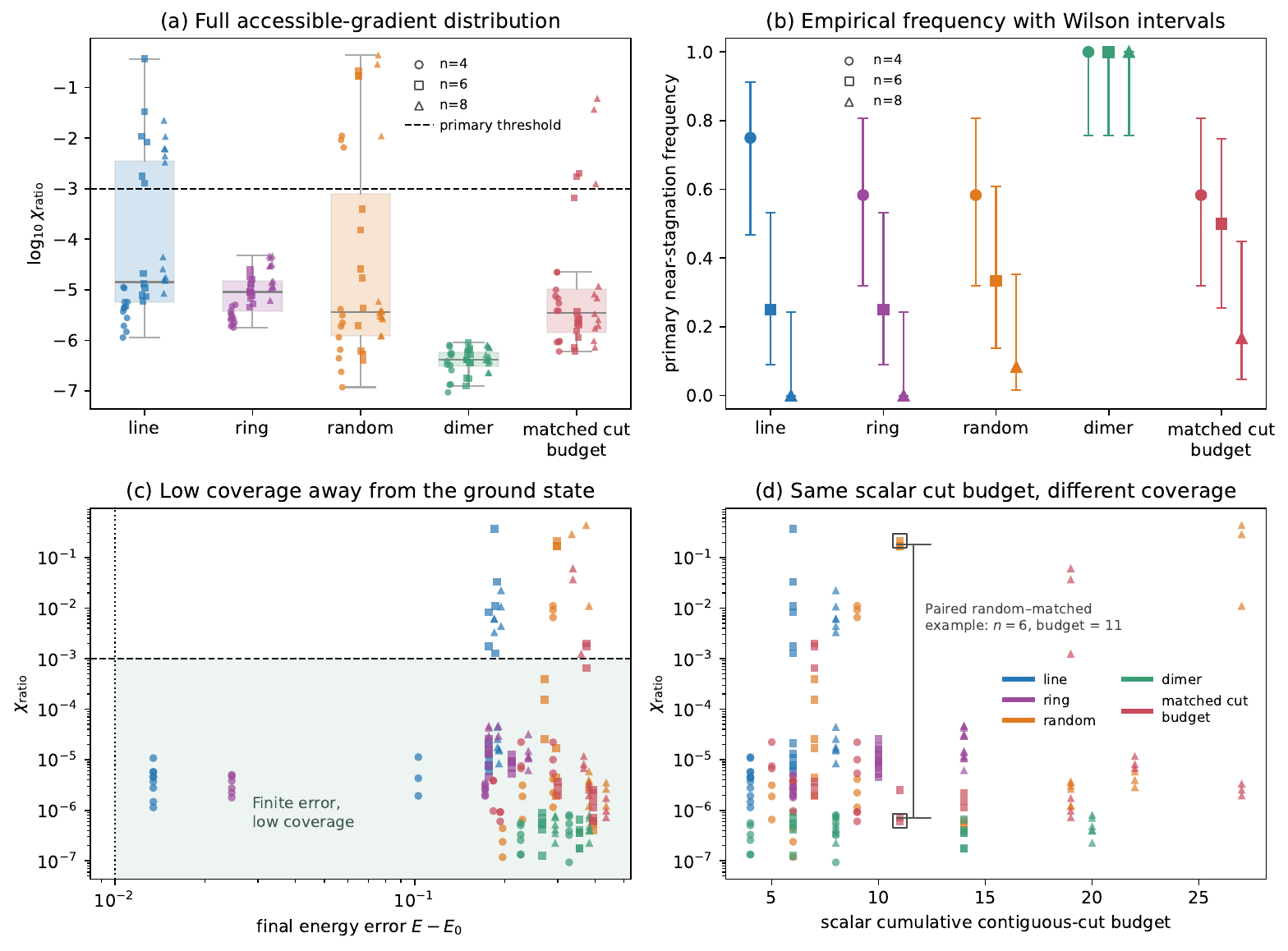}
\caption{\textbf{Near-stagnation is routing dependent across the sampled
TFIM circuits.}
For each routing family and size, four architectures are each trained
from three initializations, giving 12 runs per group and 180 runs in total.
Each scatter point in (a), (c), and (d) represents one training endpoint;
(b) reports the fraction satisfying all four conditions in
Eq.~\eqref{eq:pilot-event} within each group of 12 runs.
Colors identify routing families throughout; circles, squares, and
triangles denote $n=4,6,8$, respectively.
Matched-cut-budget controls share the scalar cumulative contiguous-cut
budget of their paired random routings, not necessarily the full cut profile.
(a) All accessible-gradient ratios.
(b) Primary-event frequencies with descriptive Wilson 95\% binomial
intervals. (c) The shaded region marks finite energy error
($E-E_0\ge10^{-2}$) with low coverage ($a\le10^{-3}$).
(d) Scalar cut capacity alone does not determine local access.
The outlined points highlight one random--matched pair at $n=6$ with
$B_{\rm cut}=11$, trained from the same initialization and evaluated
at their respective terminal points.
The highlighted random run did not meet the optimizer's success criterion.
Optimization failures are retained; intervals and conditioning diagnostics
are detailed in Appendix~\ref{app:genericity}.}
\label{fig:genericity}
\end{figure}

Together, these experiments connect routing-dependent local geometry to
geometry-guided circuit growth: same-point augmentation exposes descent directions
while preserving the learned circuit at insertion, and geometric selection
improves the use of an equal augmentation budget. The TFIM ensemble shows
that the underlying restriction of gradient access also arises beyond
the controlled construction and varies with routing.

\section{Discussion and conclusion}
\label{sec:discussion}

Viewed more broadly, our results suggest that the role of a quantum
circuit architecture may extend beyond determining what a circuit family
can ultimately represent. For a given task, an architecture must first
have sufficient representational capacity for its reachable circuit
family to accommodate good solutions; but this is only the first
threshold. Being able to accommodate good solutions does not mean that
training can readily find them. Routing constrains global
representability through structural properties such as cross-cut
capacity, while also shaping the tangent space, gradient coverage, and
conditioning near the current optimization point. A useful architecture
must therefore not only have room for good solutions, but also expose
effective directions that can lead toward them during training. In other
words, capacity constrains where one can reach in principle, while
accessibility determines which paths are actually open from one's
current position.

This distinction also suggests a broader view of variational quantum
algorithms as a problem of algorithm discovery. Hand-designed quantum
algorithms and variational circuits learned through task feedback can
be viewed as two ways of determining structure and parameters in the
same mixed discrete--continuous space $U(\bm A,\bm\theta)$. Notably,
a final answer can have a remarkably compact description even when the
space in which it must be found is enormous. A single
$A\in\GL(n,\Ftwo)$ requires only $O(n^2)$ bits to specify, and an
architecture with polynomially many routing blocks still has a
polynomial-size description; yet the corresponding discrete candidate
space contains $2^{\Theta(tn^2)}$ routing tuples. The resources needed
to describe a known quantum circuit and those needed to discover it
from the information a task provides are therefore different complexity
questions. A short quantum algorithm may exist without an equally
accessible route to finding it.

The block count $t$ brings this tension into focus.
Increasing $t$ can enlarge the representable circuit family, allowing
an increasingly rich range of quantum computations to be described in
this mixed language; at the same time, it rapidly expands the discrete
architecture space to be explored. A larger representational space
offers more possible answers, but may also make discovery harder.
The local geometry studied here suggests an alternative to global
enumeration: rather than search the entire architecture space in
advance, use the descent component not yet captured by the current
tangent space to identify which structural directions are most worth
opening next. This echoes residual-driven expansion of effective
subspaces, where information not captured by the current representation
itself guides the next enlargement. Same-point augmentation further
allows this expansion to preserve the learned physical map at insertion,
so that structural search and continuous training need not remain
separate stages \cite{Grimsley2019,Sohail2026GeoADAPT,Yao2021AVQDS}.
Training is then not only a matter of moving along existing paths:
it can also use the information supplied by the current task to decide
where the next path should open.

This picture also has an interesting parallel in the development of
classical machine learning. A highly expressive function family does
not automatically ensure that practical training can find its useful
members; architecture, parameterization, initialization, and optimization
paths jointly shape the region that can ultimately be learned
\cite{LeCun2015DeepLearning,He2016,Elsken2019}. Similarly, for
variational quantum circuits, it may be unnecessary to expect a final,
perfect architecture to be found first, leaving only its continuous
parameters to optimize. Continuous updates determine motion along the
directions already exposed, while structural updates determine which
additional directions can be opened. In comparison with classical
neural networks, some structural degrees of freedom in quantum circuits
also admit an explicit finite-field algebraic description, making the
question of how structure changes the learning geometry directly
amenable to analysis. This suggests a more dynamic picture of quantum
circuit design: architecture and parameters may grow together during
training.

The discrete--continuous representation also provides a layered
perspective on practical quantum computational advantage. A candidate
circuit family must first have sufficient representational capacity;
useful structures and parameters must then be identified through
theoretical design or a resource-bounded discovery process. Hardware
connectivity, compilation costs, and fault-tolerant implementation,
where required, further restrict the practically available choices of
$\bm A$ and $\bm\theta$. Ultimately, quantum advantage must be
assessed against appropriate task-level classical baselines, rather
than inferred solely from the difficulty of simulating the circuit.
In particular, a quantum circuit that is hard to simulate classically
does not automatically confer quantum advantage on the task it solves
\cite{Huang2026VastWorld}. From representability in principle to
discoverability, physical and fault-tolerant realizability, and finally
task-level advantage over classical methods, these considerations impose
distinct constraints on the enormous mixed circuit space. Thus, the
computational space permitted in principle by quantum mechanics and the
part of that space that can ultimately be put to use need not coincide.

This raises a further question beyond the boundaries traditionally
studied in computational complexity: among quantum circuits that are
computationally powerful in principle, how much can a resource-bounded
learning, search, or design process actually reach from the information
supplied by a task? The answer will generally depend on the available
task information and the resources permitted to the discovery process.
Routing geometry provides one concrete local view of this broader
problem: capacity, tangent accessibility, and local structural growth
make parts of this discovery boundary observable, diagnosable, and
actionable.

Advantages in computational complexity are not the only advantages
available in quantum information processing. Particularly while the
scale, depth, and reliability of near-term quantum circuits remain
resource-constrained, quantum measurement and the acquisition of
information offer another route: information security
\cite{BennettBrassard1984}, sample or experimental efficiency
\cite{Huang2026VastWorld,Huang2022Experiments}, and efficient classical
postprocessing \cite{Wang2025QuditShadow}. Quantum advantage need not
take a single form, nor be characterized by a single computational
complexity boundary.

\begin{samepage}
Perhaps, then, the question ultimately facing quantum circuit design
is not only which circuits are sufficiently powerful in principle, but
which powerful circuits can be discovered, starting from the information
a task provides, along a resource-feasible path. Understanding how these
paths emerge, how they are shaped by architecture, and what geometry
characterizes the \emph{reachable frontier} they collectively trace
may help us recognize an equally important yet less explored aspect
of quantum computational capability.

The question is not only what can be computed, but what can be
discovered from the information a task provides.
\end{samepage}

\paragraph{Acknowledgements.}
This work was supported by the National Natural Science Foundation of
China under Grants No.~42330707 and No.~42530108.
OpenAI GPT-6 was used as an assistive tool for mathematical discussion,
code development, numerical analysis, and manuscript editing.

\paragraph{Code and data availability.}
Figure-reproduction code and frozen data are available in
Ref.~\cite{Wang2026RoutingRepository}.

\clearpage
\appendix
\etocdepthtag.toc{appendix}
\begingroup
\renewcommand{\contentsname}{Appendix Contents}
\etocsettagdepth{main}{none}
\etocsettagdepth{appendix}{subsection}
\etocsetnexttocdepth{subsection}
\tableofcontents
\endgroup
\clearpage

\section{Encoding and finite-field routing}
\label{app:bqp}

Section~\ref{sec:mixed} treats routing matrices and local rotations as
coordinates for quantum circuits. This appendix establishes that
representation and its resource bounds, then derives the finite-field
construction underlying each routing block.

\subsection{Mixed representation and compilation}

Proposition~\ref{prop:mixed-coverage} connects the mixed representation
to ordinary gate circuits in both directions. We first group a given
circuit into alternating blocks, then track the description length and
compilation accuracy needed for polynomial-size circuit families.

\paragraph{Alternating form and phase convention.}
\begin{proposition}[Alternating representation]
\label{prop:alternating}
Every finite circuit over CNOT and arbitrary single-qubit unitaries can be
written in the form \eqref{eq:alternating} without changing gate order.
\end{proposition}

\begin{proof}
Group maximal consecutive CNOT and local subwords without reordering gates.
Insert identities on untouched qubits and at the endpoints when needed.
This gives Eq.~\eqref{eq:alternating}, up to the irrelevant global phase.
\end{proof}

Every $\mathrm U(2)$ gate is a scalar times an $\SU(2)$ gate; replacing
local factors by $\SU(2)$ representatives changes only the global phase.
CNOT permutations can also have determinant $-1$, so evaluation takes
values in $\PU(2^n)=\mathrm U(2^n)/\mathrm U(1)$:
\begin{equation}
 \operatorname{ev}_{n,t}:\cE_{n,t}\longrightarrow\PU(2^n),\qquad
 (\bm A,\bm L)\longmapsto[L_tC(A_t)\cdots C(A_1)L_0].
 \label{eq:evaluation-map}
\end{equation}
We write $U$ for a representative when no phase-dependent quantity is used.
The Euler layers in Eq.~\eqref{eq:local-layer} give the redundant
coordinate space
\begin{equation}
 \cP_{n,t}=\bigl(\GL(n,\Ftwo)\bigr)^t\times\mathbb R^{3n(t+1)}.
 \label{eq:mixed-space}
\end{equation}
For projective evaluation, each angle may be reduced modulo $2\pi$,
giving the representatives used in Eq.~\eqref{eq:explicit-coordinate-values}.

\paragraph{Resource bounds and coverage.}
For a polynomial bound $p$, define
\begin{equation}
 \cE_{n,p}=\bigcup_{0\le t\le p(n)}\cE_{n,t},
 \label{eq:polynomial-universe}
\end{equation}
and its Euler-coordinate counterpart
\begin{equation}
 \cP_{n,p}=\bigcup_{0\le t\le p(n)}\cP_{n,t}.
 \label{eq:polynomial-coordinate-universe}
\end{equation}
When accuracy varies, take $p=p(n,\lceil\log_2(1/\epsilon)\rceil)$.
Recording each angle to $b_\theta$ bits gives description length
\begin{equation}
 O\!\left(tn^2+3n(t+1)b_\theta\right).
 \label{eq:encoding-length}
\end{equation}
Gaussian elimination synthesizes each binary map with $O(n^2)$ CNOTs
under all-to-all logical connectivity \cite{PatelMarkovHayes08};
the construction and its physical implementation costs are detailed in
Section~\ref{app:cnot-details}.
Set $d_{\rm PU}([U],[V])=\min_\varphi\|U-e^{i\varphi}V\|_{\rm op}$.

\begin{corollary}[Polynomial-size unitary circuit coverage]
\label{cor:polynomial-coverage}
Let $\{Q_x\}$ be a uniform family of $n=n(|x|)$ qubit unitary circuits of
polynomial size in $|x|$ over a fixed finite, effectively specified, bounded-arity universal gate
library with effectively compilable gates.  We allow idle-qubit padding so that $n(|x|)\ge |x|$; thus polynomial
bounds in input length are also polynomial in $n$.  For
every $0<\epsilon<1$, standard compilation produces
$e_x\in\cE_{n,t}$ such that
\begin{equation}
 d_{\rm PU}\!\left(\operatorname{ev}_{n,t}(e_x),[Q_x]\right)
 \le\epsilon,
 \qquad
 t=\operatorname{poly}\!\left(n,\log\frac1\epsilon\right),
 \label{eq:coverage-bound}
\end{equation}
with angle precision
$b_\theta=\operatorname{poly}(n,\max\{1,\log(1/\epsilon)\})$.  In particular, under
the usual inverse-polynomial or exponentially small accuracy conventions,
the logical unitary skeleton lies in a polynomial-resource union
\eqref{eq:polynomial-universe}.
\end{corollary}

\begin{proof}
Compile the bounded-arity gates into CNOTs and one-qubit gates
\cite{Barenco1995,DawsonNielsen2006,Nielsen2010}.
Allocating error $\epsilon/N$ per factor in a compiled word of length $N$
bounds the total operator error by $\epsilon$ by telescoping.
Standard synthesis keeps gate count and angle precision polynomial in
$n$ and $\log(1/\epsilon)$.  Grouping the compiled word gives no more
routing blocks than gates, proving Eq.~\eqref{eq:coverage-bound}.
\end{proof}

For the converse in Proposition~\ref{prop:mixed-coverage}, synthesize each
of the $t$ binary maps by Gaussian elimination using $O(n^2)$ CNOTs.
Each of the $n(t+1)$ local factors has a three-rotation Euler
decomposition, giving the stated gate counts.  Approximate compilation of
these rotations into a fixed discrete universal library adds the standard
precision-dependent overhead; with polynomial angle precision and the
stated accuracy conventions, this overhead remains polynomial.

\paragraph{Scope of the representation.}
The construction encodes a supplied logical unitary circuit, including
the unitary core of standard BQP computations \cite{BernsteinVazirani1997}.
It covers polynomial-size circuit descriptions such as Shor factoring,
while Hamiltonian simulation and HHL retain their usual assumptions on
input preparation, evolution time, condition number, precision, and output
access \cite{Shor1997,Lloyd1996,harrow2009quantum}. For Grover search,
the black-box query count is exponential in the address width
\cite{Grover1996}; encoding a supplied circuit does not remove its oracle
calls. Measurement, reset, classical feedback, general channels, and
physical error correction require additional process and resource
descriptions. Encoding a specified circuit is therefore distinct from
discovering a task-suitable circuit.

The coordinate representation also fixes the local geometry used below.
Euler coordinates are periodic and can be singular at special points,
so reachable images need not be embedded manifolds globally.
Accordingly, we define the accessible tangent space as the image of
the evaluation differential at the chosen parameter preimage,
transported to the identity. Different preimages of the same unitary
may therefore yield different differential images.

\subsection{CNOT circuits as finite-field transformations}
\label{app:cnot-details}
On $n$ qubits there are $n(n-1)$ directed CNOT gates. Each adds the
control bit to the target bit modulo two and therefore acts as an
elementary invertible linear transformation over $\Ftwo$.
The group they generate is the linear-reversible subgroup of the
Clifford group, naturally isomorphic to $\GL(n,\Ftwo)$.
Consequently, any CNOT-only circuit is represented by a single matrix
$A\in\GL(n,\Ftwo)$, and conversely every such $A$ can be synthesized
using CNOT gates.

We now make this correspondence explicit and derive its synthesis,
counting, and Pauli-propagation rules. Binary vectors are written as
columns; the rightmost matrix or gate in a product acts first. We identify
$\ket{x_1\cdots x_n}$ with $x=(x_1,\ldots,x_n)^{\mathsf T}\in\Ftwo^n$.

\subsubsection{From XOR to an elementary matrix}
A CNOT with control $i$ and distinct target $j$ acts as
\begin{equation}
 x_j\longmapsto x_j\oplus x_i,\qquad
 x_k\longmapsto x_k\quad(k\ne j),\qquad
 A_{i\to j}=I_n+E_{ji},
 \label{eq:cnot-elementary}
\end{equation}
where $E_{ji}$ has its only nonzero entry in row $j$, column $i$ and all
arithmetic is over $\Ftwo$.  For example,
\[
 A_{1\to2}=
 \begin{pmatrix}1&0&0\\1&1&0\\0&0&1\end{pmatrix},\qquad
 A_{2\to3}A_{1\to2}=
 \begin{pmatrix}1&0&0\\1&1&0\\1&1&1\end{pmatrix}.
\]
Thus applying $1\to2$ followed by $2\to3$ maps
$(x_1,x_2,x_3)$ to $(x_1,x_2+x_1,x_3+x_2+x_1)$.
Because $E_{ji}^2=0$ and $1+1=0$,
$(I+E_{ji})^2=I$: every elementary CNOT is its own inverse.

\subsubsection{Composition and faithful representation}
For $A\in\GL(n,\Ftwo)$, define $C(A)\ket{x}=\ket{Ax}$ and extend by
linearity.  Invertibility permutes an orthonormal basis, so $C(A)$ is unitary.
On every computational-basis vector,
\begin{equation}
 C(A_2)C(A_1)\ket{x}=\ket{A_2A_1x},\qquad
 C(A_2)C(A_1)=C(A_2A_1),\qquad
 C(A)^\dagger=C(A^{-1}).
 \label{eq:cnot-composition-details}
\end{equation}
If $C(A)=C(B)$, applying the maps to all binary unit vectors gives $A=B$.
The representation is therefore faithful.  A time-ordered gate list
$T_1,\ldots,T_m$ implements binary matrix $T_m\cdots T_1$.
Different lists can have the same product, for instance by insertion of
$T^2=I$.  The logical map is unique; a physical synthesis is not.

\subsubsection{Generation and Gaussian elimination synthesis}
Gaussian elimination reduces an invertible binary matrix to $I$ using row
additions and row swaps.  Row scalings are unnecessary because the only
nonzero scalar in $\Ftwo$ is one.  A row swap itself uses three additions:
\[
 (a,b)\mapsto(a,a+b)\mapsto(b,a+b)\mapsto(b,a).
\]
Thus row additions alone generate $\GL(n,\Ftwo)$, including permutation
matrices.  In circuit language these three additions are the familiar
three-CNOT SWAP synthesis.

More explicitly, suppose Gaussian elimination records elementary
transformations $T_m\cdots T_1A=I$. Since each CNOT transformation is
self-inverse, $T_j^{-1}=T_j$, reversing the recorded operations gives
\[
 A=T_1\cdots T_m,
\]
and hence a CNOT synthesis of $C(A)$. A straightforward elimination
procedure uses $O(n^2)$ elementary row additions, including the three
additions for each row swap, giving an $O(n^2)$ CNOT construction.
The classical elimination requires $O(n^3)$ bit operations without
packing, so synthesis time and quantum gate count are distinct resources.
More refined all-to-all synthesis methods reduce the asymptotic CNOT
count to $O(n^2/\log n)$ \cite{PatelMarkovHayes08}.

\paragraph{Logical blocks and physical implementation.}
The block count $t$ counts logical CNOT maps $C(A_\ell)$ separated by
local one-qubit layers; it is not a hardware-depth measure.
A single logical map $A_\ell$ may require many CNOT gates, and its
physical depth depends on connectivity, native gates, and scheduling.
Our grouping preserves the original gate order and every intervening
one-qubit layer. Different CNOT decompositions of the same $A_\ell$
therefore realize the same ideal logical unitary $C(A_\ell)$, while
their circuit depth and noisy implementations may differ.

\subsubsection{How many logical routing choices are there?}
Count linearly independent columns successively.  After choosing $k$
independent columns their span contains $2^k$ vectors, so the next column
has $2^n-2^k$ choices.  Therefore
\begin{equation}
 |\GL(n,\Ftwo)|
 =\prod_{k=0}^{n-1}(2^n-2^k)
 =2^{n^2}\prod_{j=1}^{n}(1-2^{-j})
 =\Theta(2^{n^2}).
 \label{eq:gl-order-details}
\end{equation}
The last product converges to a positive constant, since
$\sum_{j\ge1}2^{-j}<\infty$ and the logarithms are summable after any
finite initial segment. Thus each specified matrix admits polynomial-size
synthesis, while exhaustive enumeration of the candidate routing space
is exponential. This separates the cost of implementing a chosen routing
from that of selecting one for a task.

\subsubsection{Pauli propagation}
To connect the routing coordinate to the tangent generators of
Section~\ref{sec:tangent}, we now follow its action on Pauli operators.

For binary Pauli labels define
$X^x=\bigotimes_q X_q^{x_q}$ and $Z^z=\bigotimes_q Z_q^{z_q}$.
For any computational-basis vector $\ket y$, direct action gives
\begin{align*}
 C(A)X^xC(A)^\dagger\ket y
 &=C(A)X^x\ket{A^{-1}y}\\
 &=C(A)\ket{A^{-1}y+x}
 =\ket{y+Ax}=X^{Ax}\ket y,\\
 C(A)Z^zC(A)^\dagger\ket y
 &=C(A)Z^z\ket{A^{-1}y}\\
 &=(-1)^{z^{\mathsf T}A^{-1}y}\ket y
 =(-1)^{(A^{-\mathsf T}z)^{\mathsf T}y}\ket y
 =Z^{A^{-\mathsf T}z}\ket y.
\end{align*}
Since this holds on every basis vector,
\begin{equation}
 C(A)X^xC(A)^\dagger=X^{Ax},\qquad
 C(A)Z^zC(A)^\dagger=Z^{A^{-\mathsf T}z}.
 \label{eq:binary-pauli-forward}
\end{equation}
Replacing $A$ by $A^{-1}$ gives the body-frame pullback:
\[
 C(A)^\dagger X^xC(A)=X^{A^{-1}x},\qquad
 C(A)^\dagger Z^zC(A)=Z^{A^{\mathsf T}z}.
\]
General Hermitian Pauli strings also carry a sign/phase convention, which
must be retained when computing signed derivatives.  These formulas explain
how a binary routing coordinate controls the propagated generator dictionary.

\paragraph{CZ routing coordinates.}
CZ-only blocks give a second binary routing coordinate. There are
$\binom n2$ CZ gate types, encoded by a symmetric zero-diagonal matrix
$\Gamma\in\Ftwo^{n\times n}$ through
$U_\Gamma=\prod_{i<j}\operatorname{CZ}_{ij}^{\Gamma_{ij}}$.
The gates commute and square to identity, so
\[
 U_{\Gamma_1}U_{\Gamma_2}=U_{\Gamma_1+\Gamma_2},\qquad
 G_{\rm CZ}\cong\Ftwo^{\binom n2},\qquad
 |G_{\rm CZ}|=2^{\binom n2},
\]
an Abelian elementary $2$-group, compared with
$|\GL(n,\Ftwo)|=2^{n^2-O(1)}$.
Conjugation by $\operatorname{CZ}_{ij}$ sends
$X_i\mapsto X_iZ_j$, $X_j\mapsto Z_iX_j$ and fixes $Z_i,Z_j$.
Thus $U_\Gamma$ propagates binary Pauli labels as
$(x,z)\mapsto(x,z+\Gamma x)$, with the Hermitian Pauli sign convention
retained for signed derivatives.

Across $S\mid\bar S$, write
\[
 \Gamma=\begin{pmatrix}\Gamma_{SS}&B\\B^{\mathsf T}&\Gamma_{\bar S\bar S}\end{pmatrix},
 \qquad B=\Gamma_{S,\bar S}.
\]
Within-side CZ gates are local to this cut. The remaining diagonal
operator has coefficient matrix $((-1)^{a^{\mathsf T}Bb})_{a,b}$ in the
orthonormal diagonal-operator bases $\{\ket a\bra a\}$ and $\{\ket b\bra b\}$.
Its distinct rows are the $2^{\rank_{\Ftwo}B}$ orthogonal binary
characters, giving
\begin{equation}
 \OSR_{S\mid\bar S}(U_\Gamma)=2^{\rank_{\Ftwo}B}.
 \label{eq:cz-osr}
\end{equation}
The symmetric CZ coordinate has one cross-adjacency block, whereas
the CNOT formula uses the two directed blocks $Q,R$ through
$\kappa_A(S)=\rank Q+\rank R$.

The identity
$\operatorname{CNOT}_{i\to j}=(I_i\otimes H_j)\operatorname{CZ}_{ij}(I_i\otimes H_j)$
relates the two universal gate libraries with arbitrary local gates.
Gate-by-gate conversion may add local Hadamard layers and change the
block count, but preserves polynomial-size descriptions.
Thus the discrete-routing viewpoint is not specific to CNOT blocks:
CZ blocks provide a simpler Abelian example in which a binary graph
coordinate controls both cross-cut capacity and the propagated tangent
dictionary.

\section{Reachability, capacity, and fidelity proofs}

The global restrictions in Section~\ref{sec:osr} follow from three
steps: identify reachable-set relations, compute the capacity of one
routing block across a cut, and propagate that capacity through a
multiblock circuit. The final rank--fidelity argument translates these
structural bounds into a task error.

\subsection{Reachable-set relations and the exact CNOT cut formula}
\label{app:osr-proof}

We first establish when local rotations make two routings equivalent.
The exact cut formula then provides a capacity invariant for
distinguishing their reachable sets.

\paragraph{Proof of Proposition~\ref{prop:single-block-images}.}
If the images intersect, local factors satisfy
\[
 [L_1C(A)L_0]=[K_1C(B)K_0],\qquad
 L_0,L_1,K_0,K_1\in\mathcal L_n.
\]
Rearranging gives $[C(B)]=[K_1^\dagger L_1C(A)L_0K_0^\dagger]$.
Closure of $\mathcal L_n$ under products and inverses then identifies the
two full images.  Conversely, this local equivalence immediately makes
the images equal.  Thus any nonempty intersection implies equality.
For the CNOT direction example, $H$ can be represented in $\SU(2)$ up to
global phase, consistent with the projective convention.

\paragraph{Rotation propagation and Hadamard equivalence.}
Conjugation by $C=\mathrm{CNOT}_{c\to q}$ fixes $Z_c$ and $X_q$ and sends
$X_c$ to $X_cX_q$ and $Z_q$ to $Z_cZ_q$.
The identity $C e^{-i\theta P/2}C^\dagger
=e^{-i\theta CPC^\dagger/2}$, with $C^\dagger=C$, proves
Eqs.~\eqref{eq:local-rotation-commutation}--\eqref{eq:rotation-propagation}.
The addition rule for same-axis rotations follows because their
generators commute.

Let $F=H^{\otimes n}$.  Its binary Fourier action gives
\[
 FC(A)F|x\rangle
 =2^{-n}\sum_{y,z\in\Ftwo^n}
 (-1)^{x^{\mathsf T}y+z^{\mathsf T}Ay}|z\rangle
 =|A^{-\mathsf T}x\rangle.
\]
The sum over $y$ vanishes unless $A^{\mathsf T}z=x$, proving
Eq.~\eqref{eq:hadamard-inverse-transpose}.
For an elementary CNOT matrix $E=I+e_qe_c^{\mathsf T}$,
$E^{-\mathsf T}=I+e_ce_q^{\mathsf T}$, which swaps control and target.
For a chronological word $A=E_m\cdots E_1$,
$A^{-\mathsf T}=E_m^{-\mathsf T}\cdots E_1^{-\mathsf T}$;
the chronological order is unchanged.  This differs from inversion,
which reverses the order of the self-inverse CNOT factors.

\paragraph{Overlapping but unequal multiblock images.}
For two qubits, the tuple $(I,I)$ generates exactly the local unitaries.
Let $C=C(A)=\mathrm{CNOT}_{1\to2}$.  The tuple $(A,A)$ evaluates to
$L_2CL_1CL_0$.  Setting $L_1=I$ and using $C^2=I$ produces every local
unitary.  Setting $L_2=L_0=I$ and
$L_1=I\otimes R_z(\phi)$ instead gives
\[
 C(I\otimes R_z(\phi))C=\exp(-i\phi Z_1Z_2/2).
\]
At $\phi=\pi/2$, this maps $|++\rangle$ to an entangled state (its
coefficient matrix has nonzero determinant), so it is not a product
unitary.  Hence $\mathcal M_{(I,I)}\subsetneq\mathcal M_{(A,A)}$.
Thus multiblock reachable sets can overlap without being equal, unlike
the single-block images of Proposition~\ref{prop:single-block-images}.

\paragraph{Proof of Theorem~\ref{thm:osr-formula}.}
Using the blocks $P,Q,R,T$ of Theorem~\ref{thm:osr-formula}, write
basis indices as $x=(u,v)$ with $u\in\Ftwo^{|S|}$ and
$v\in\Ftwo^{|\barS|}$.  Vectorize the operator and group the output and input
registers belonging to the same side of the physical cut:
\begin{equation}
\lvert C(A)\rangle\!\rangle
=
\sum_{u,v}
\ket{Pu+Qv,u}_{L}
\otimes
\ket{Ru+Tv,v}_{R}.
\label{eq:vectorized-A}
\end{equation}
The operator Schmidt rank of $C(A)$ equals the Schmidt rank of this vector
across $L\mid R$.

Apply on $L$ the reversible basis permutation
$\ket{a,u}\mapsto\ket{a+Pu,u}$ and on $R$ the permutation
$\ket{b,v}\mapsto\ket{b+Tv,v}$.  Both operations are local with respect to
$L\mid R$, so they preserve Schmidt rank and transform
\eqref{eq:vectorized-A} into
\[
\sum_{u,v}\ket{Qv,u}_{L}\otimes\ket{Ru,v}_{R}.
\]
After reordering registers within $L$ and within $R$, this vector factorizes
as
\begin{equation}
\left(\sum_u\ket{u}_{L_u}\ket{Ru}_{R_R}\right)
\otimes
\left(\sum_v\ket{Qv}_{L_Q}\ket{v}_{R_v}\right).
\label{eq:two-channel-factorization}
\end{equation}
For the first factor, group terms by $w=Ru$:
\[
\sum_u\ket{u}\ket{Ru}
=
\sum_{w\in\im R}
\left(\sum_{u:Ru=w}\ket{u}\right)\ket{w}.
\]
The preimage sets for distinct $w$ are disjoint, so the left vectors are
mutually orthogonal; the right computational-basis vectors are also
orthogonal.  The Schmidt rank is therefore
$|\im R|=2^{\rank_{\Ftwo}R}$.  The same argument gives
$2^{\rank_{\Ftwo}Q}$ for the second factor.  Schmidt ranks multiply under
tensor product, proving
\[
\OSR_{S\mid\barS}(C(A))
=2^{\rank_{\Ftwo}Q+\rank_{\Ftwo}R}.
\]

The same vector is obtained, up to normalization, by applying $C(A)$
to one half of $n$ Bell pairs. In this Choi-state description, the
grouped cut is $(S_{\rm out}S_{\rm ref})\mid
(\barS_{\rm out}\barS_{\rm ref})$. Applying Hadamards to all output
qubits gives the graph state with amplitudes
$2^{-n}(-1)^{y^{\mathsf T}Ax}$ on $|y\rangle_{\rm out}|x\rangle_{\rm ref}$.
Its output--reference adjacency block is $A$, and its cross-adjacency
matrix for the grouped physical cut is
$\left(\begin{smallmatrix}0&Q\\R^{\mathsf T}&0\end{smallmatrix}\right)$.
The graph-state cut-rank relation~\cite{Hein2004GraphStates} therefore
also yields $2^{\rank_{\Ftwo}Q+\rank_{\Ftwo}R}$. This application of
the graph-state relation distinguishes the two directed blocks of $A$
within an undirected graph on output and reference qubits. The direct
factorization above keeps the physical qubit labels fixed and applies
whether or not the diagonal blocks $P$ and $T$ are invertible.

\paragraph{Cost of a specified-cut calculation.}
Let $k=|S|$ and $s=\min\{k,n-k\}$.  Each off-diagonal block has
$k(n-k)$ entries and rank at most $s$.  Ordinary binary Gaussian
elimination uses at most $s$ pivot stages, each requiring
$O(k(n-k))$ bit operations, yielding Eq.~\eqref{eq:cut-rank-cost}.
Extracting the two blocks costs $O(k(n-k))$ when the binary entries of
$A$ are already accessible; storing the full matrix uses $n^2$ bits.
The integer OSR can be reported as $2^{\kappa_A(S)}$ or in binary with
at most $2s+1$ bits.

When the input is a gate list, first construct $A$.
For a list of $g$ CNOTs, start from $I_n$ and update the target row by
XOR with the control row for each gate.  This uses $O(n^2+gn)$ bit
operations and $O(n^2)$ storage. Input processing therefore scales with
the supplied list length $g$, including when that list is exponentially
long in $n$. A full labeled profile also has its own output cost:
enumerating all nontrivial cuts up to complementation requires
$2^{n-1}-1$ entries.

\subsection{OSR submultiplicativity and cumulative budgets}
\label{app:cumulative-proof}

The single-block formula bounds a full circuit through
submultiplicativity. Applying the same argument to individual gates
then relates the cumulative budget to physical depth on a chain.

If $X=\sum_{a=1}^{r}L_a\otimes R_a$ and
$Y=\sum_{b=1}^{s}\widetilde L_b\otimes\widetilde R_b$, then
\[
XY=\sum_{a,b}(L_a\widetilde L_b)\otimes(R_a\widetilde R_b),
\]
so $\OSR(XY)\le rs$.  Iterating this inequality through
\eqref{eq:alternating}, using OSR one for all local layers and
Theorem~\ref{thm:osr-formula} for every CNOT block, proves
\eqref{eq:cumulative}.  The state-rank statement follows similarly by
applying each operator Schmidt term to each state Schmidt term.

Together with the rank--fidelity inequality proved in
Appendix~\ref{app:rank-fidelity-proof}, this gives the loss bound for
a target maximally entangled on $m$ qubit pairs: any output with
$K_{\bm A}(S)<m$ obeys
\[
\cL\ge1-2^{K_{\bm A}(S)-m}.
\]
The bound supplies a necessary cut budget. Attaining it also depends on
the intermediate local layers: routing blocks can cancel, and local
angles can suppress channels. A target with demands on several cuts
requires those demands to be met together.

\paragraph{Proof of Corollary~\ref{cor:shallow-chain}.}
In a specified CNOT implementation, a gate contained within one side of a
cut has OSR one, while a crossing CNOT has OSR two.
Submultiplicativity therefore gives $\OSR(U)\le 2^{g_S}$.
Equivalently, for each routing block $\kappa_{A_\ell}(S)$ is no larger
than the crossing CNOT count in that block, so $K_{\bm A}(S)\le g_S$.
For a nearest-neighbor chain, a consecutive-site cut has one boundary
edge, usable at most once in each layer of disjoint gates.
Thus $g_S\le d$.  A product input has Schmidt rank one, giving the
output rank bound.  Entropy is at most the base-two logarithm of Schmidt
rank, proving Eq.~\eqref{eq:shallow-state}.

A pure target consisting of $m$ cross-cut Bell pairs (and fixed local
pure states on any remaining qubits) has $2^m$ equal squared Schmidt
coefficients.  The rank--fidelity inequality of
Appendix~\ref{app:rank-fidelity-proof} therefore bounds its squared
overlap with an output of rank at most $2^{g_S}$ by
$F\le\min\{1,2^{g_S-m}\}\le\min\{1,2^{d-m}\}$.
This proves the claimed preparation obstruction.
On an all-to-all device, disjoint crossing CNOTs can instead act in one
layer, and each maps $|+\rangle|0\rangle$ to a Bell pair.
Both constructions refer to the same fixed labeled bipartition.

For a product input, the chain bound also applies at every intermediate
step.  If $d=O(\log n)$, all consecutive-cut Schmidt ranks remain
polynomially bounded, permitting matrix-product-state simulation with
polynomial resources under standard efficiently evaluable gate and
finite-precision assumptions \cite{Vidal2003,MarkovShi2008}.
This simulation consequence uses the consecutive-cut rank bounds of
the nearest-neighbor chain.

\subsection{An explicit multiblock operator with OSR three}
\label{app:osr-three}

The cumulative budget bounds the final OSR, while intermediate local
rotations can change its exact value. In particular, a product of CNOT
blocks separated by non-Clifford local gates can have OSR three, even
though each CNOT block has power-of-two OSR. A three-qubit example across
$1\mid23$ is
\begin{equation}
U'=
\operatorname{CNOT}_{1\to2}
\,[R_x(\pi/3)_1]\,
\operatorname{CNOT}_{1\to3}
\,[R_x(\pi/3)_1]\,
\operatorname{CNOT}_{1\to2}.
\label{eq:osr-three-circuit}
\end{equation}
Writing the right factors on qubits $2,3$ and suppressing identity factors
on untouched qubits, direct multiplication gives
\begin{equation}
\begin{split}
 U'={}&I_1\otimes\frac{I+X_3}{4}
      +Z_1\otimes\frac{I-X_3}{2}\\
     &-i\frac{\sqrt3}{4}X_1\otimes X_2(I+X_3).
\end{split}
\label{eq:osr-three-decomposition}
\end{equation}
The three right factors, including their scalar coefficients, are pairwise
Hilbert--Schmidt orthogonal: $(I+X_3)(I-X_3)=0$, and the mixed traces with
$X_2(I+X_3)$ vanish because $\Tr X_2=0$.  Their squared norms are
$1/2,2,3/2$, while each left Pauli has squared norm $2$.
Thus the nonzero operator Schmidt singular values are $2,\sqrt3,1$, and
$\OSR_{1\mid23}(U')=3$.  Direct realignment of this displayed circuit
independently gives these three nonzero singular values.

Three-qubit unitaries of multipartite Schmidt rank three and their
CNOT-based implementations have been studied previously
\cite{Song2021}. Equation~\eqref{eq:osr-three-circuit}
establishes the chosen-cut value $\OSR_{1\mid23}=3$ directly within the
alternating circuit representation used here.

\subsection{Rank--fidelity inequality}
\label{app:rank-fidelity-proof}

To turn a Schmidt-rank ceiling into the Bell-task obstruction, we bound
the overlap of a rank-limited state with the maximally entangled target.

Let the Schmidt coefficients of a normalized pure state $\ket{\psi}$ across the
target cut be $\lambda_1,\ldots,\lambda_R$.  After optimizing local Schmidt
bases, Cauchy--Schwarz gives
\[
\abs{\braket{\Phi_m}{\psi}}
\le 2^{-m/2}\sum_{j=1}^{R}\lambda_j
\le 2^{-m/2}\sqrt{R\sum_j\lambda_j^2}
=\sqrt{\frac{R}{2^m}}.
\]
Squaring proves Lemma~\ref{lem:rank-fidelity}.  For $R\le2^m$, equality holds when the
nonzero Schmidt coefficients are uniform and aligned with $R$ target Schmidt
vectors.  Applying this bound to the normalized pure product input and
interior minimizer of Theorem~\ref{thm:bell-floor} gives its stated
obstruction; the nonzero-ambient conclusion uses purity explicitly.

\section{Routing dictionaries and tangent geometry}

Sections~\ref{sec:tangent} and~\ref{sec:trainability} distinguish the
directions exposed by routing from the capacity of the reachable family.
We first illustrate that distinction with Pauli dictionaries, then derive
the gradient identities, stability bounds, and local fidelity response
used in the main text.

\subsection{A same-profile, different-dictionary counterexample}
\label{app:profile-counterexample}

The cut profile of Section~\ref{sec:osr} records cross-cut capacity,
whereas local accessibility depends on the propagated generators.
The examples below distinguish these two kinds of information: first,
compressing a profile loses the placement of capacity; then, even a
complete labeled profile can leave the Pauli dictionary undetermined.
The underlying bi-cutrank uses both off-diagonal blocks of a directed
invertible binary map, as distinct from graph cut-rank for an undirected
graph \cite{KanteRao2013}.

\paragraph{Maximum and unlabeled-profile counterexamples.}
For three qubits in cut order
$(1\mid23,2\mid13,12\mid3)$, a single
$\operatorname{CNOT}_{1\to2}$ has profile $(1,1,0)$, whereas the sequence
$\operatorname{CNOT}_{1\to2}$ followed by
$\operatorname{CNOT}_{1\to3}$ has profile $(1,1,1)$.  Both have maximum OSR
two, so the maximum alone loses capacity placement.  For four qubits in cut
order $(1,2,12,3,13,23,123)$, the matrices
\[
A_0=
\begin{pmatrix}
0&0&0&1\\0&0&1&0\\0&1&0&0\\1&0&0&0
\end{pmatrix},
\quad
B_0=
\begin{pmatrix}
0&0&1&0\\0&0&0&1\\1&0&0&0\\0&1&0&1
\end{pmatrix}
\]
have profiles $(2,2,4,2,4,0,2)$ and $(2,2,4,2,0,4,2)$.
Their sorted multisets are both $(0,2,2,2,2,4,4)$, but the strongest
capacities occur on different labeled cuts.

\paragraph{Identical profile, different Pauli dictionary.}
Let $\mathcal D_A^{\rm pull}$ denote the real span of the local Pauli
operators pulled back through the single routing block $C(A)$:
\[
 \mathcal D_A^{\rm pull}
 =\Span_{\mathbb R}\{C(A)^\dagger P_qC(A):
      P\in\{X,Y,Z\},\ q=1,\ldots,n\}.
\]
For four qubits, consider
\begin{equation}
A=
\begin{pmatrix}
0&0&0&1\\
0&0&1&0\\
0&1&0&0\\
1&0&0&0
\end{pmatrix},
\qquad
B=
\begin{pmatrix}
0&0&0&1\\
0&0&1&0\\
0&1&0&0\\
1&0&0&1
\end{pmatrix}.
\label{eq:profile-counterexample-matrices}
\end{equation}
Using the same seven nonredundant labeled cuts as above gives
\[
q(A)=q(B)=(2,2,4,2,4,0,2).
\]
Nevertheless,
\[
C(A)^\dagger X_1C(A)=X_4,\qquad
C(B)^\dagger X_1C(B)=X_1X_4,
\]
and direct comparison of all 12 local Pauli images gives
$X_4\in\mathcal D_A^{\rm pull}\setminus\mathcal D_B^{\rm pull}$ and
$X_1X_4\in\mathcal D_B^{\rm pull}\setminus\mathcal D_A^{\rm pull}$.

The reason is that $\kappa_A(S)$ records the dimensions of off-diagonal
binary images, while Pauli propagation depends on the images of
individual coordinate vectors. Concretely, for any routing matrix $A$,
CNOT conjugation acts on binary Pauli
labels $(x,z)\in\Ftwo^{2n}$ as
\begin{equation}
(x,z)\stackrel{C(A)(\cdot)C(A)^\dagger}{\longmapsto}
(Ax,A^{-\mathsf T}z).
\label{eq:symplectic-cnot}
\end{equation}
Thus two matrices can agree on all cut ranks yet disagree on the columns of
$A^{-1}$ or $A^{\mathsf T}$. Writing $P(x,z)$ for the Hermitian Pauli
string with binary labels $(x,z)$, up to an overall sign, the pulled-back
dictionary is
\begin{equation}
\mathcal D_A^{\rm pull}=
\Span_{\mathbb R}\{
P(A^{-1}e_q,0),P(0,A^{\mathsf T}e_q),P(A^{-1}e_q,A^{\mathsf T}e_q):q=1,\ldots,n
\}.
\label{eq:dictionary-appendix}
\end{equation}
Thus identical cut profiles can coexist with different propagated
directions, which are the inputs to the local accessibility analysis.

\subsection{From measurements to tangent geometry}
\label{app:first-principles}
Section~\ref{sec:tangent} evaluates a parameter derivative both through
shifted expectation values and through its overlap with the ambient
gradient. Here we derive the two descriptions of that same derivative,
then relate their normalization to the state-space metric. We use the loss
$\cL(\theta)=\Tr(O U(\theta)\rho U(\theta)^\dagger)$ with fixed density
operator $\rho$ and fixed Hermitian observable $O$.  Infidelity to a fixed
pure target differs from such an expectation only by an additive constant.

\subsubsection{Expectation measurements and the parameter-shift rule}
For one independently parameterized Pauli rotation, factor the circuit as
\[
 U(\theta)=U_{\rm after}e^{-i\theta P/2}U_{\rm before},\qquad P^2=I,
\]
where the surrounding gates do not depend on this parameter.  Set
$\rho_{\rm pre}=U_{\rm before}\rho U_{\rm before}^\dagger$ and
$O_{\rm post}=U_{\rm after}^\dagger O U_{\rm after}$.
The Pauli identity implies
\[
 e^{-i\theta P/2}=\cos(\theta/2)I-i\sin(\theta/2)P.
\]
Substitution into
$\Tr(O_{\rm post}e^{-i\theta P/2}\rho_{\rm pre}e^{i\theta P/2})$
gives $a+b\cos\theta+c\sin\theta$, where
\begin{align*}
 a&=\tfrac12\Tr[O_{\rm post}(\rho_{\rm pre}+P\rho_{\rm pre}P)],\\
 b&=\tfrac12\Tr[O_{\rm post}(\rho_{\rm pre}-P\rho_{\rm pre}P)],\\
 c&=\tfrac{i}{2}\Tr[O_{\rm post}(\rho_{\rm pre}P-P\rho_{\rm pre})].
\end{align*}
The coefficients are real.  Differentiating the sinusoid, or taking the
half-difference of two shifted sinusoids, proves the standard identity
\begin{equation}
 \partial_\theta\cL
 =\frac12\left[\cL(\theta+\pi/2)-\cL(\theta-\pi/2)\right].
 \label{eq:parameter-shift-details}
\end{equation}
This recovers the standard Pauli-rotation parameter-shift rule of
Mitarai et al.~\cite{Mitarai2018} and Schuld et al.~\cite{Schuld2019}.

Experimentally one estimates the two expectation values using repeated
measurements and takes their half-difference.  If independent sample means
have variances $\sigma_+^2/N_+$ and $\sigma_-^2/N_-$, the derivative
estimator is unbiased with variance
$\tfrac14(\sigma_+^2/N_++\sigma_-^2/N_-)$. The two shifted expectations
are therefore estimated from repeated shots; a Pauli-sum observable can
also require several measurement settings and grouping choices.

The formula applies to one independently parameterized gate occurrence.
For a shared parameter, sum the occurrence-wise derivatives by the chain
rule. For a nonlinear function of measured expectations, apply the
classical chain rule to the expectation derivatives.

\subsubsection{From a unitary velocity to a Hermitian generator}
Differentiate $U^\dagger U=I$ to obtain
$\dot U^\dagger U+U^\dagger\dot U=0$.  Hence
\begin{equation}
 H=iU^\dagger\dot U=H^\dagger,\qquad \dot U=-iUH.
 \label{eq:body-generator-details}
\end{equation}
This is left translation to the identity, expressed in Hermitian rather
than anti-Hermitian coordinates.

For the single gate factored above, direct multiplication yields
\begin{equation}
 H=\tfrac12 U_{\rm before}^\dagger P U_{\rm before}.
 \label{eq:body-prefix-details}
\end{equation}
It is the \emph{prefix} (gates before the parameter) that enters this
body-frame generator.  The output-frame generator is instead
$\widetilde H=i\dot U U^\dagger
=\tfrac12U_{\rm after}P U_{\rm after}^\dagger=UHU^\dagger$,
which uses the suffix.  Keeping the frames distinct avoids reversing Pauli
propagation in a multilayer circuit.

The global-phase convention determines the trace of $H$. If $U$ lies
in $\SU(d)$, Jacobi's determinant formula gives
$\Tr(U^\dagger\dot U)=\partial_\theta\log\det U=0$,
so $\Tr H=0$. The same conclusion holds for any constant-determinant
curve. Our fixed-routing circuits with special-unitary local factors
have constant determinant even when a CNOT permutation has determinant
$-1$. For a general $\mathrm U(d)$ curve with varying global phase,
$H$ need not be traceless; its identity component is removed in the
projective description and has no effect on the expectation loss.

\subsubsection{The intrinsic gradient and the measured derivative}
Let $O_U=U^\dagger O U$.  Along velocity $-iUH$,
\begin{align}
 d\cL_U[-iUH]
 &=-i\Tr(O_U[H,\rho])\nonumber\\
 &=i\Tr([O_U,\rho]H)
 =\inner{G(U)}{H},\qquad G(U)=i[O_U,\rho].
 \label{eq:gradient-sign-details}
\end{align}
The second line follows because
$\Tr(O_U[H,\rho])=-\Tr([O_U,\rho]H)$.
The commutator is anti-Hermitian and traceless, so $G$ is Hermitian and
traceless.  For the unnormalized Hilbert--Schmidt Riemannian metric,
$\operatorname{grad}\cL(U)=-iUG(U)$.  Thus $G$ is the Hermitian
representation of the ascent gradient, and the direction represented by
$-G$ is descent \cite{Absil2008,Stokes2020}.

The parameter-shift derivative and \eqref{eq:gradient-sign-details} are
two evaluations of the same scalar.  The first is a measurement protocol;
the second separates the objective-dependent gradient $G(U)$ from the
circuit-dependent direction $H$.  Estimating derivatives does not require
tomographic reconstruction of the full ambient operator.

For several parameters, the same calculation gives
$g_j=\inner{G}{H_j}$. Proposition~\ref{prop:gram} then converts these
frame inner products into the accessible squared gradient norm
$g^{\mathsf T}M^+g$.

\subsubsection{Pure-state variance and metric normalization}
The operator and state-space gradients describe the same loss with
different metrics. Their relation for a pure input fixes the
normalization used in the numerical diagnostics.
For $\rho=\ket{\psi_0}\bra{\psi_0}$, write
$a=\bra{\psi_0}O_U\ket{\psi_0}$ and
$\ket{w}=(O_U-a)\ket{\psi_0}$, so $\braket{\psi_0}{w}=0$.
Then
\[
 [O_U,\rho]=\ket{w}\bra{\psi_0}-\ket{\psi_0}\bra{w},
 \qquad
 \norm{G}_{\rm HS}^2=2\braket{w}{w}
 =2\operatorname{Var}_{\psi_0}(O_U).
\]
The factor two follows from the two orthogonal off-diagonal components.
For the horizontal state metric the gradient is
$2(O-\langle O\rangle)\ket{\psi}$, whose squared norm is
$4\operatorname{Var}_{\psi}(O)$. Thus the chosen metric determines
which normalization to use.
Zero variance means an eigenstate, which need not be the minimizing
eigenstate.  A bounded observable supplies an upper bound on the gradient,
not a uniform positive lower bound during training.

\subsection{Bell-state metric consistency}
\label{app:numerics}

The Bell calculation in Fig.~\ref{fig:bell-floor} compares the ambient and
accessible gradient norms using the same pure-state horizontal metric.
We therefore evaluate the deterministic $(m,k)$ constructions in this
metric. With $O=-\ket{\Phi_m}\bra{\Phi_m}$, the calculation uses
$G_{\rm st}=2(O-\langle O\rangle)\ket\psi$ and horizontal tangents $v_j$.
Both displayed norms are therefore state-space quantities:
\[
 \|G_{\rm st}\|^2=4F(1-F),\qquad
 \|\Pi_{\cT^{\rm st}}G_{\rm st}\|^2
 =g^{\mathsf T}(M^{\rm st})^+g.
\]
Direct statevector derivatives and an independent SVD projection onto
the tangent range give the same accessible-gradient norm.
The state Schmidt rank is computed from the reshaped statevector, and
the capacity-saturating constructions attain the bound in
Eq.~\eqref{eq:bell-floor} to machine precision.

\subsection{Perturbative stability of near-stagnation}
\label{app:near-stagnation}

The local bound in Section~\ref{sec:obstruction} asks how much visible
gradient can appear when both the state and its tangent space change.
We first bound the coverage ratio in terms of projector and gradient
perturbations, then give a singular-value condition that controls the
projector perturbation.

\begin{proposition}[Perturbative near-stagnation]
\label{prop:robustness}
Let $P$ and $P'$ be orthogonal tangent-space projectors in the same Hilbert
inner product.  Suppose
\[
 PG=0,\quad \norm{G}\ge g_0>0,\quad
 \norm{P'-P}_{\rm op}\le\delta_P,\quad
 \norm{G'-G}\le\delta_G.
\]
Then Eq.~\eqref{eq:robust-projection} holds.  If
$\delta_G<\norm{G}$, then
\begin{equation}
 \frac{\norm{P'G'}}{\norm{G'}}
 \le
 \frac{\delta_P\norm{G}+\delta_G}{\norm{G}-\delta_G}.
 \label{eq:robust-ratio}
\end{equation}
\end{proposition}

\begin{proof}
Since $PG=0$,
$P'G'=P'(G'-G)+(P'-P)G$.  The first bound follows from
$\norm{P'}_{\rm op}\le1$ and the triangle inequality.  The denominator in
Eq.~\eqref{eq:robust-ratio} follows from
$\norm{G'}\ge\norm{G}-\delta_G$.
\end{proof}

To obtain the projector bound from tangent generators, let the columns of $V$ and
$V'=V+E$ be tangent generators.  If both matrices have rank $r$,
$\sigma_r(V)\ge\gamma>0$, and
$\eta=\norm{E}_{\rm op}<\gamma$, pseudoinverse perturbation gives the
conservative bound
\begin{equation}
 \norm{\Pi_{\im V'}-\Pi_{\im V}}_{\rm op}
 \le \frac{\eta}{\gamma-\eta}+\frac{\eta}{\gamma}
 \le \frac{2\eta}{\gamma-\eta}.
 \label{eq:projector-gap}
\end{equation}
Thus a stationary point with a locally constant tangent rank and a positive
minimum nonzero singular-value gap has a neighborhood in which
$\norm{G}$ remains bounded below and
$\norm{\Pi_{\cT}G}/\norm{G}$ remains arbitrarily small. The resulting
neighborhood has small accessible gradients; exact stationarity is
required only at its reference point. Constant rank and a positive gap
are essential here: when the rank changes or the gap closes, an
arbitrarily small generator perturbation can introduce a finite new
subspace.

\subsection{Local fidelity expansion and cross-Gram response}
\label{app:frame-details}

The local target-overlap analysis in Section~\ref{sec:trainability}
uses fidelity to compare the displacements generated by two routing
frames. We derive its quadratic expansion and cross-Gram response,
then relate the matched-frame Hessian to the CNOT-origin spectrum.

The origin of each circuit in Theorem~\ref{thm:origin-frame} represents
$\ket{0^n}$.  Choose smooth local phases so first derivatives are horizontal,
and write the first-order state displacements as $u=V_{\bm B}\theta$ and
$w=V_{\bm A}\theta^\star$. Normalization fixes the real reference-state
component of each second-order displacement, giving
\begin{align*}
 \langle\psi_{\bm B}(\theta)|\psi_{\bm A}(\theta^\star)\rangle
 ={}&1+\langle u|w\rangle-\tfrac12\norm u^2-\tfrac12\norm w^2
       +i\varphi^{(2)}\\
 &+O\bigl((\norm\theta+\norm{\theta^\star})^3\bigr),
\end{align*}
where $\varphi^{(2)}$ is a real second-order phase term.
Taking the squared modulus removes this phase to second order, yielding
\[
 1-|\langle\psi_{\bm B}(\theta)|\psi_{\bm A}(\theta^\star)\rangle|^2
 =\norm{u-w}_{\mathbb R}^2
   +O\bigl((\norm{\theta}+\norm{\theta^\star})^3\bigr).
\]
With adjoints and norms taken in the real state-space metric, the
quadratic surrogate expands as
\begin{align*}
 \cL_{\bm B}^{(2)}(\theta)
 &=\norm{V_{\bm B}\theta-V_{\bm A}\theta^\star}_{\mathbb R}^2\\
 &=\theta^{\mathsf T}M_{\bm B}^{\rm st}\theta
   -2\theta^{\mathsf T}K_{\bm B,\bm A}\theta^\star
   +(\theta^\star)^{\mathsf T}M_{\bm A}^{\rm st}\theta^\star,
\end{align*}
using $M_{\bm B}^{\rm st}=V_{\bm B}^*V_{\bm B}$ and
$K_{\bm B,\bm A}=V_{\bm B}^*V_{\bm A}$.
Differentiating this homogeneous quadratic polynomial gives
\[
 \nabla_\theta\cL_{\bm B}^{(2)}
 =2M_{\bm B}^{\rm st}\theta-2K_{\bm B,\bm A}\theta^\star.
\]
For matched routings, $\bm B=\bm A$ implies
$K_{\bm A,\bm A}=M_{\bm A}^{\rm st}$ and hence
\[
 \nabla_\theta^2\cL_{\bm A}^{(2)}=2M_{\bm A}^{\rm st}.
\]
Theorem~\ref{thm:origin-frame} places the positive Gram eigenvalues in
$[1/4,\ell/4]$, so the positive quadratic Hessian eigenvalues lie in
$[1/2,\ell/2]$. For the quadratic
surrogate, conditioning is therefore certified on the identifiable
quotient, with null directions retained in the full parameter space.
Smooth Taylor remainders exist for each fixed circuit, with constants
that may depend on the circuit and system size.

\paragraph{Extension to Clifford suffixes.}
The same origin calculation extends to Clifford suffixes and Pauli slots
provided the complete origin circuit still prepares $\ket{0^n}$ up to
global phase.  Write $Q=\sigma i^{x\cdot z}X^xZ^z$ for a Hermitian Pauli.
On $\ket{0^n}$ it produces a computational-basis label and a phase.
Under this reference-state condition, horizontalization removes $x=0$;
the remaining real tangents are classified
by $(x,\text{real/imaginary quadrature})$ and an orientation sign.
Equal labels in the same quadrature form rank-one Gram blocks.
The CNOT-only $R_zR_xR_z$ theorem is the special case with only the
imaginary quadrature and unit orientation for the $X$ slots.
For a general stabilizer reference $\ket{\psi_{\rm ref}}$, the tangent is
$-\tfrac{i}{2}(I-\ket{\psi_{\rm ref}}\bra{\psi_{\rm ref}})Q
\ket{\psi_{\rm ref}}$; the binary condition $x=0$ alone does not determine
whether it vanishes.

\section{Numerical protocols and statistical checks}

The experiments in Section~\ref{sec:repair} test same-point recovery,
geometric selection, and routing-dependent near-stagnation. The following
subsections give their protocols and supporting statistics in the same
order: the controlled Bell test, the connected Heisenberg task and its
complementary TFIM selector control, then the TFIM routing ensemble.

\subsection{Controlled multiblock protocol and hierarchical statistics}
\label{app:statistics}

The controlled test in Section~\ref{sec:controlled} compares candidates
at a common trained point, so differences in subsequent descent can be
related to the directions they add. Here we specify the candidate pools
and summarize score--descent correlations and selector performance
across instances (Fig.~\ref{fig:multiblock}).

\subsubsection{Protocol}

The controlled test uses $(n,t)=(4,2),(6,3),(8,4)$ with five fixed
seeds per configuration. Targets contain $m=n/2$ Bell pairs, non-Clifford
local layers, and within-side routing. The old ansatz has no cross-cut CNOT
and reaches product-state infidelity floors $0.75$, $0.875$, and $0.9375$;
its state distances from $\ket{0^n}$ lie in $[0.729,1.521]$.
Each instance has 12 nonidentity candidates and an identity control.
Every gadget adds eight CNOTs and $3n$ angles, with 50 Adam steps.
The target routing is absent from the pool, and every candidate satisfies
$\kappa_B(S)<m$ on the target cut. Comparisons include identity,
random, OSR-only, candidate-only, incremental, and retrospective
finite-pool oracle selection.

The identity control uses $B=I$ and still adds the $3n$ local angles;
cancelling CNOT pairs retain the same eight gate slots. The Bell OSR-only
selector maximizes the target-informed score
\[
 \sum_{S\in\mathcal S_{\rm task}}
 \min\{\kappa_B(S),d_S\}(1+0.1d_S),
\]
where $d_S$ is the rounded binary logarithm of the target Schmidt rank
across $S$, computed with relative singular-value threshold $10^{-10}$.
The set $\mathcal S_{\rm task}$ contains up to ten nonredundant cuts,
prioritizing larger $d_S$ and then more balanced cuts, with the target cut
always included. OSR-only ties are resolved by the lowest candidate index.

\subsubsection{Per-instance correlation summaries}

For each of the 15 instances, all 12 nonidentity candidates share the same
target, old point, ambient gradient, and loss floor.  They are therefore
repeated measurements within an instance; the statistical unit is the instance.
Spearman correlations are computed separately within each instance.
Uncertainty for pooled standardized slopes is obtained by resampling complete
instances with replacement.

Figure~\ref{fig:multiblock-rank-detail} retains the candidate-level
20-step comparison, complementing the five-step per-instance summary
in Fig.~\ref{fig:multiblock}(c).

\begin{figure}[!htbp]
\centering
\includegraphics[width=.68\linewidth]{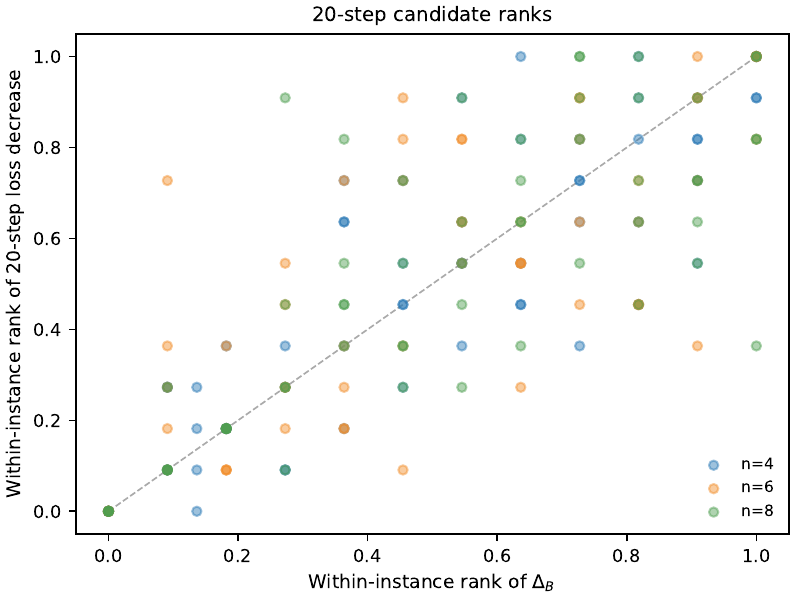}
\caption{\textbf{Candidate-level score--descent ranks at 20 steps.}
The 180 candidate--instance pairs comprise 12 nonidentity candidates
in each of 15 instances. Both ranks are computed within each instance
and normalized to $[0,1]$, using average ranks for ties; coincident points
overlap. The dashed diagonal indicates identical rankings. Correlations
are evaluated separately within each instance, as summarized below.}
\label{fig:multiblock-rank-detail}
\end{figure}

The table reports the median, interquartile range (IQR), and full range
of the per-instance correlations at each training horizon.

\begin{center}
\begin{tabular}{lrrrr}
\toprule
score & steps & median $\rho$ & IQR & min--max\\
\midrule
$\Delta_B^{\rm st}$ & 1 & .916 & [.804,.951] & [.517,.986]\\
$\Delta_B^{\rm st}$ & 5 & .902 & [.857,.941] & [.797,1.000]\\
$\Delta_B^{\rm st}$ & 20 & .839 & [.773,.909] & [.364,.958]\\
$\Delta_B^{\rm st}$ & 50 & .534 & [.343,.598] & [.147,.923]\\
candidate-only & 1 & .921 & [.804,.931] & [.510,.986]\\
candidate-only & 5 & .902 & [.832,.941] & [.755,1.000]\\
candidate-only & 20 & .839 & [.759,.906] & [.364,.937]\\
candidate-only & 50 & .503 & [.346,.583] & [.147,.944]\\
OSR-only & 1 & .062 & [-.063,.404] & [-.446,.808]\\
OSR-only & 5 & .107 & [-.003,.372] & [-.187,.895]\\
OSR-only & 20 & .158 & [-.016,.391] & [-.235,.662]\\
OSR-only & 50 & .222 & [.112,.390] & [-.058,.764]\\
\bottomrule
\end{tabular}
\end{center}

Within-instance rank-standardized slopes for $\Delta_B^{\rm st}$ at
$1,5,20,50$ steps are $.858,.898,.809,.489$, with instance-cluster bootstrap
95\% intervals $[.789,.916]$, $[.868,.927]$, $[.722,.879]$, and
$[.379,.598]$. The corresponding candidate-only slopes are nearly identical.
These intervals quantify variation across the 15 constructed instances.

\subsubsection{Paired 50-step comparisons}

Method comparisons use the 50-step loss decrease and are paired by
instance. W/T/L denotes wins, ties, and losses for the first selector
relative to the second. The random comparator in this table is one fixed
draw from each instance's nonidentity candidate pool.

\begin{center}
\begin{tabular}{lrrr}
\toprule
contrast & mean difference & bootstrap 95\% interval & W/T/L\\
\midrule
incremental $-$ random & .1017 & [.035,.177] & 13/0/2\\
incremental $-$ OSR-only & .0175 & [-.029,.067] & 11/0/4\\
incremental $-$ candidate-only & $1.11\times10^{-6}$ & $[-1.02,4.34]\times10^{-6}$ & 1/13/1\\
candidate-only $-$ random & .1017 & [.036,.178] & 13/1/1\\
\bottomrule
\end{tabular}
\end{center}
Wilcoxon and sign tests are treated as exploratory small-sample diagnostics
without multiplicity correction.

The saved nonidentity pools give mean probabilities
$0.850$, $0.333$, and $0.150$ of drawing a candidate within 90\% of the
finite-pool oracle gain at $n=4,6,8$, respectively. Each probability averages
the exact pool fractions over five instances; an additional random baseline
uses 1000 draws per instance, without retraining.
Candidate-only and incremental selection have
paired win/tie/loss $1/13/1$, giving comparable performance in this test.

\subsection{Heisenberg-chain protocol and statistics}
\label{sec:state-geometry}

The Heisenberg test in Section~\ref{sec:xxz} compares selectors at a
stalled checkpoint under equal augmentation and retraining budgets.
We describe the state-space score calculation, the checkpoint and
candidate protocol, and the paired statistics behind Fig.~\ref{fig:xxz}.

\subsubsection{State-space implementation}

The simulations use the horizontal state tangents and real inner product
defined in Section~\ref{sec:trainability}. For the energy objective
$E=\bra{\psi}H\ket{\psi}$, the ambient state gradient is
\begin{equation}
 G_{\rm st}=2(H-E)\ket{\psi}.
 \label{eq:state-energy-gradient}
\end{equation}
To implement the real metric with complex statevectors, the code forms
\begin{equation}
 V_{\mathbb R}=\begin{pmatrix}\operatorname{Re}V\\
 \operatorname{Im}V\end{pmatrix},
 \qquad
 G_{\mathrm{st},\mathbb R}=\begin{pmatrix}\operatorname{Re}G_{\rm st}\\
 \operatorname{Im}G_{\rm st}\end{pmatrix}.
 \label{eq:realified-state-geometry}
\end{equation}
Then $M^{\rm st}=V_{\mathbb R}^{\mathsf T}V_{\mathbb R}$ and
$g=V_{\mathbb R}^{\mathsf T}G_{\mathrm{st},\mathbb R}$.
Candidate-only scores use the SVD range of the candidate state tangents.
Incremental scores residualize both those tangents and $G_{\rm st}$
against the old state-tangent space before projection. A direct
union-space calculation independently checks the same increment.

\subsubsection{Checkpoint and candidate protocol}

For each $n=8,10,12$ and each of five seeds, the baseline starts from a small
random perturbation of a N\'eel product state.  It is optimized for 140 Adam
steps followed by at most 260 L-BFGS-B iterations and 24 small-gradient
refinement steps.  Scanning this trajectory without candidate information,
the code accepts the first point whose last-20-step gain is at most
\begin{equation}
 \tau=\max\{2\times10^{-7},2\times10^{-4}(E-E_0)\},
 \label{eq:xxz-checkpoint-threshold}
\end{equation}
whose 20-step energy variation is at most $5\tau$, whose parameter-gradient
norm is at most $2\times10^{-3}$, and whose energy error exceeds
$\max\{10^{-3},10^{-3}(n-1)\}$.  All 15 fixed instances pass this audit.

Each candidate sequence is a random permutation of the $n-1$ nearest-neighbor
edges with independently chosen directions, subject to uniqueness of its
binary routing matrix.  The sequence is used once in $C(B)$ and in reverse
for $C(B)^\dagger$.  Candidate generation is seed-fixed and ground-state
blind.  Every augmented candidate receives 50 Adam updates at the same
learning-rate schedule.  The no-augmentation curve is a continuation control
with no added resources; random, candidate-only, incremental, and oracle
comparisons all refer to candidates with exactly $3n$ added angles and
$2(n-1)$ added CNOTs.

Recovery requires the 50-step gain to exceed continuation of the old ansatz
by at least one percent of the checkpoint energy error, with an absolute
floor of $10^{-5}$.
Across the accepted checkpoints, energy errors lie in $[2.286,3.095]$,
parameter-gradient norms in $[1.111\times10^{-3},1.875\times10^{-3}]$,
and ambient state-gradient norms in $[5.721,7.145]$.

\subsubsection{Paired statistics}
\label{app:xxz-statistics}

Candidate-level Spearman correlations are computed separately inside each of
the 15 fixed pools. Method differences are paired by instance. The reported
95\% intervals are percentile intervals from 200,000 resamples of the 15
complete instances with replacement, using seed 20260804. Exact two-sided
sign tests remove ties with absolute difference at most $10^{-7}$.
The first table gives selector-level gains, recovery rates, and oracle
regrets; the second gives the paired contrasts.

\begin{center}
\begin{tabular}{lrrr}
\toprule
method & mean 50-step gain & recovery rate & mean oracle regret\\
\midrule
no augmentation & $-0.000059$ & $0/15$ & $1.015$\\
fixed random candidate & $0.495$ & $15/15$ & $0.520$\\
candidate-only & $0.925$ & $15/15$ & $0.090$\\
incremental & $0.916$ & $15/15$ & $0.099$\\
finite-pool oracle & $1.015$ & $15/15$ & $0$\\
\bottomrule
\end{tabular}
\end{center}

The fixed-random contrast compares against a single seeded draw from the
candidate pool, independent of the candidate scores and training gains.
The pool-mean contrast uses all already-trained candidates to remove
draw noise from the random baseline.

\begin{center}
\begin{tabular}{lrrr}
\toprule
paired contrast & mean difference & bootstrap 95\% interval & W/T/L\\
\midrule
candidate-only $-$ fixed random & $.4297$ & $[.2409,.5961]$ & 13/1/1\\
incremental $-$ fixed random & $.4210$ & $[.2538,.5773]$ & 13/1/1\\
candidate-only $-$ pool mean & $.3838$ & $[.2569,.4969]$ & 13/0/2\\
incremental $-$ pool mean & $.3750$ & $[.2675,.4790]$ & 15/0/0\\
incremental $-$ candidate-only & $-.0088$ & $[-.0805,.0573]$ & 2/11/2\\
\bottomrule
\end{tabular}
\end{center}
The corresponding exact sign-test values are $.00183$, $.00183$, $.00739$,
$6.10\times10^{-5}$, and $1$, respectively.

Incremental scoring changes the selected candidate in $4/15$ instances.
The median within-pool Spearman correlation with 50-step gain is $0.713$
for incremental and $0.559$ for candidate-only scoring. The maximum
principal-angle obliqueness is $44.98^\circ$. Here obliqueness is
$\max_k\min\{\alpha_k,90^\circ-\alpha_k\}$ for principal angles
$\alpha_k$ between the old and candidate tangent spaces: it measures the
largest departure from parallel or orthogonal directions. The largest
projector-commutator Frobenius norm is $2.674$.
These subspace diagnostics explain the changes in selection; the paired
table above reports the resulting optimization gains.

\paragraph{Numerical verification.}
The two-qubit Hamiltonian check gives the exact singlet energy $-3$. The
analytic energy gradient agrees with centered finite differences to maximum
absolute error $1.98\times10^{-9}$ and relative error
$3.85\times10^{-9}$. The adjoint parameter gradient and the contraction of
the horizontal state tangents with $G_{\rm st}$ agree to
$4.44\times10^{-14}$. Across all 180 candidates, the maximum saved
same-point differences in state, energy, and old gradient are zero in double
precision.

\subsection{TFIM same-point control}
\label{app:tfim-control}

This complementary selector control for Section~\ref{sec:xxz} compares
candidate-only and incremental selection on a second many-body objective
and relates their near-equivalence to the compatible-geometry regime of
Theorem~\ref{thm:equivalence} (Fig.~\ref{fig:tfim}). It uses the
open-chain TFIM introduced in Section~\ref{sec:genericity},
\begin{equation}
 H_n=-\sum_{q=0}^{n-2}Z_qZ_{q+1}-\sum_{q=0}^{n-1}X_q
 \label{eq:tfim}
\end{equation}
at $J=h=1$ for $n=4,6,8$, with five fixed seeds per size.

The old ansatz has two repeated-dimer CNOT blocks separated and surrounded by
local rotations applied in time order $X\to Y\to Z$.
It can optimize within each dimer but has
no direct inter-dimer entangler.
Each of 16 nontrivial candidates is
generated by sampling four directed CNOT slots uniformly from all directed
qubit pairs and removing duplicate binary matrices, without access to the
exact ground state. Each adds eight CNOTs and $3n$ parameters and receives
50 Adam steps. The identity baseline uses the cost-matched $B=I$ control
defined in Appendix~\ref{app:statistics}. Here OSR-only maximizes
$\sum_{j=1}^{n-1}\kappa_B(S_j)$ over contiguous cuts
$S_j=\{0,\ldots,j-1\}$, resolving ties by the lowest candidate index;
this score does not use the exact ground state. Method comparisons and
repeated-random resampling use the saved candidate pool.

At the constrained optima, the energy error
lies in $[0.225818,0.294384]$, the ambient state-gradient norm $\|G_{\rm st}\|$ in
$[1.13019,1.25754]$, and
\[
\chi_{\rm st}=\frac{\norm{\Pi_{\cT^{\rm st}_{\rm old}}G_{\rm st}}^2}{\norm{G_{\rm st}}^2}
\in[2.422\times10^{-27},1.114\times10^{-25}].
\]
Here $\chi_{\rm st}=a^2$ is the squared state-space coverage: the old
ansatz sees almost none of the remaining ambient descent direction.

\begin{figure}[!htbp]
\centering
\includegraphics[width=.98\linewidth]{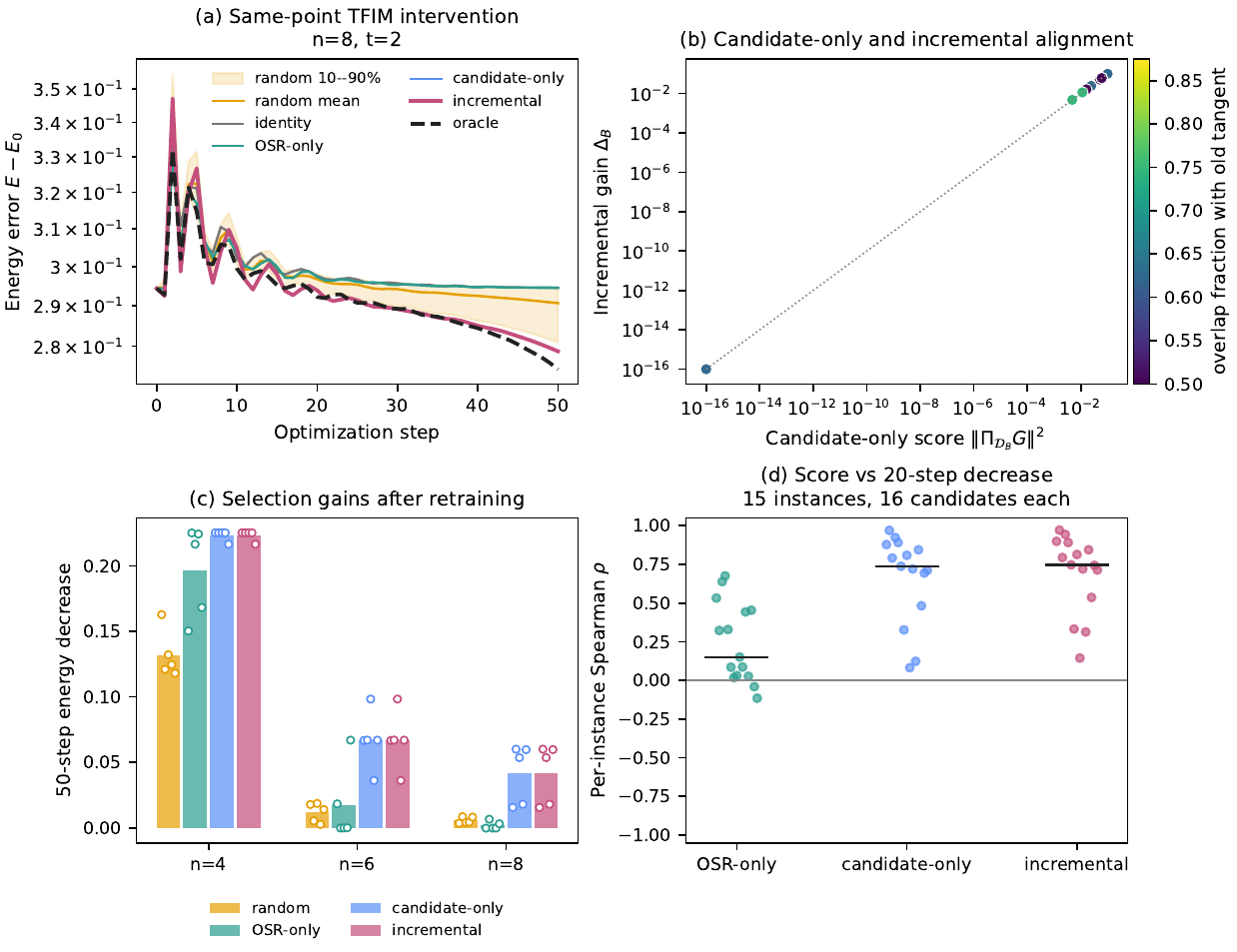}
\caption{\textbf{Complementary TFIM same-point selector control.}
The routing-conditioned augmentation is compared
across 15 independent instances (five seeds at each $n=4,6,8$) and 16
nontrivial candidates per instance.  Each candidate adds $3n$ parameters and
eight CNOTs.  Bars are five-instance means and open circles are individual
instances; panel (a) uses $n=8$, $t=2$, seed 20268723. Its random curve
and shaded band show the pointwise mean and 10--90\% range across the
16 nonidentity candidate trajectories. In (b), color gives the overlap
fraction $\sum_k\cos^2\alpha_k/r_{\rm new}$, where $\alpha_k$ are
the principal angles between the old and candidate tangent spaces and
$r_{\rm new}$ is the candidate tangent rank. Candidate-only and incremental
selection choose
the same candidate and give the same reported gain in every instance.}
\label{fig:tfim}
\end{figure}

At these same points, candidate-only and incremental scores nearly
coincide (Fig.~\ref{fig:tfim}(b)), selecting the same candidate in all
15 instances and giving identical 50-step gains at the reported precision.
Both selectors exceed the repeated-random mean in every instance.
The median probability that a random candidate attains at least 90\%
of the oracle gain is $0.073$, with range $[0.051,0.523]$.
The table summarizes the selector outcomes.

\begin{center}
\begin{tabular}{lrrr}
\toprule
selector & mean 50-step gain & median final error & median random percentile\\
\midrule
identity & -.000123 & .266037 & 6.2\\
random mean & .049730 & .251864 & 50.0\\
OSR-only & .071906 & .265804 & 56.2\\
candidate-only & .110610 & .198956 & 93.8\\
incremental & .110610 & .198956 & 93.8\\
oracle & .111452 & .198685 & 100.0\\
\bottomrule
\end{tabular}
\end{center}

Figure~\ref{fig:tfim}(d) tests shorter-horizon rank prediction at 20 steps;
panel (c) and the table report gains over the full 50-step budget.
The median within-instance Spearman correlations at 50 steps are $0.136$ for
OSR-only, $0.909$ for candidate-only, and $0.926$ for incremental scoring.
Small score differences and near ties can change rank correlations
without changing the selected candidate.

The geometric diagnostics explain why the selectors agree in this
control. Across all 240 nonidentity candidates, $\Delta_B^{\rm st}$ and
$s_{\rm cand}^{\rm st}(B)$ agree to within $1.56\times10^{-15}$ in absolute value;
principal angles deviate from
$0^\circ$ or $90^\circ$ by at most $3.195\times10^{-6}$ degrees; and the
largest Frobenius norm of the projector commutator is
$1.475\times10^{-7}$.  These data are numerically consistent with the commuting-projector
equivalence of Theorem~\ref{thm:equivalence} within the reported numerical
tolerance. In this control, candidate-only scoring captures the useful
alignment and incremental scoring gives the same selection.

\paragraph{Numerical verification.}
The maximum same-point differences in state, energy, old gradients, and
ambient gradient are zero to stored double precision.
The analytic gradient was checked against centered finite differences, with
maximum absolute error $2.44\times10^{-9}$ and maximum relative error
$1.54\times10^{-9}$.

\subsection{Preregistered TFIM genericity pilot}
\label{app:genericity}

The ensemble in Section~\ref{sec:genericity} holds the TFIM task fixed
and varies routing to measure the incidence of near-stagnation
(Fig.~\ref{fig:genericity}). Each fixed architecture is trained
independently on the open-chain Hamiltonian in Eq.~\eqref{eq:tfim},
using the common
$L_2C(A_2)L_1C(A_1)L_0$ template.
Every local layer applies rotations in time order $X\to Y\to Z$ on each qubit;
every CNOT block
contains $n/2$ gates. At each $n=4,6,8$, four routings are fixed in
each of five families.

Line/brickwork uses nearest-neighbor edges on the open chain,
with shifted edge selections between blocks; ring uses alternating
nearest-neighbor pairings on the periodic chain; random fixed-budget
routings sample an ordered control--target pair independently for each
CNOT slot; repeated dimers repeat the same disjoint pairing and CNOT
orientation in both blocks.

The fifth family, matched-cut-budget controls, consists of distinct
routing tuples paired with the random routings. To construct each pair,
random tuples with the same CNOT count are sampled until their
$B_{\rm cut}$, defined in Section~\ref{sec:genericity}, equals that of
the paired random routing, with duplicate tuples rejected. This matches
the scalar sum of the cumulative contiguous-cut budgets while allowing
the detailed cut profile to differ.

Each routing is optimized from three fixed small perturbations of
$\ket{+}^{\otimes n}$ using 120 Adam steps followed by at most 300
L-BFGS-B iterations, with stopping tolerances
$\texttt{gtol}=10^{-9}$ and $\texttt{ftol}=10^{-13}$.
The geometric diagnostics and event criteria are evaluated at the
terminal parameters returned by L-BFGS-B, giving 180 endpoints in total.

The preregistered primary event uses the state-space coverage
$\chi_{\rm ratio}=a$ of Section~\ref{sec:trainability}, with the
ambient gradient evaluated by Eq.~\eqref{eq:state-energy-gradient},
and requires
\begin{equation}
\norm{\nabla_\theta E}\le10^{-6},\qquad
\norm{G_{\rm st}}\ge10^{-2},\qquad
\chi_{\rm ratio}\le10^{-3},\qquad
E-E_0\ge10^{-2}.
\label{eq:pilot-event}
\end{equation}
These conditions identify small parameter gradients despite a nonzero
ambient descent direction, little access to that direction, and a finite
energy error.

The primary event occurs in 85/180 instances and at least one initialization
of 39/60 routings; all 85 events have optimizer-reported convergence.
The table below gives the primary-event counts by routing family and size.

\begin{center}
\begin{tabular}{lrrr}
\toprule
routing family & $n=4$ & $n=6$ & $n=8$\\
\midrule
line/brickwork & 9/12 & 3/12 & 0/12\\
ring & 7/12 & 3/12 & 0/12\\
random fixed budget & 7/12 & 4/12 & 1/12\\
repeated dimer & 12/12 & 12/12 & 12/12\\
matched cut budget & 7/12 & 6/12 & 2/12\\
\bottomrule
\end{tabular}
\end{center}

\paragraph{Reproducibility checks.}
The preregistered relaxed, primary, and strict criteria identify
151/180, 85/180, and 3/180 events, respectively. The primary count uses
all instances, including 28 runs for which L-BFGS-B did not report
convergence; none of these 28 runs satisfies the primary criterion.
State-space ranks and condition numbers are computed from retained
tangent singular values, whereas parameter-gradient norms use the
actual coordinate derivatives. In a post hoc conditioning check, 60 of
the 85 primary-event endpoints, spanning 27 routings, have retained
state-Gram condition numbers at most $10^{14}$.

Figure~\ref{fig:genericity} reports descriptive Wilson intervals,
while an additional routing-level bootstrap groups the three
initializations of each routing. Analytic gradients and tangent
contractions agree to $1.281\times10^{-14}$.
All classification thresholds were fixed before the run.
A subsequent plotting correction affected only visualization and
postprocessing and did not change any classified instance.

\FloatBarrier
\begingroup
\small
\bibliographystyle{unsrt}
\bibliography{refs}
\endgroup

\end{document}